\documentclass[letterpaper,12pt]{amsart}

\usepackage{latexsym, amsmath, amssymb,amsthm, natbib,verbatim, xy}
\xyoption{all}

\usepackage{tgpagella}
\usepackage{mathpazo}

\usepackage[left=1in,top=1.125in,right=1in,bottom=1.125in]{geometry}
\usepackage{setspace,color}
\usepackage{graphicx}
\usepackage{enumerate}
\usepackage[multiple]{footmisc}
\usepackage[leqno]{amsmath}
\usepackage{hyperref}
\usepackage{etoolbox}
\patchcmd{\section}{\scshape}{\bfseries}{}{}
\makeatletter
\renewcommand{\@secnumfont}{\bfseries}
\makeatother

\usepackage{tikz}
\newcommand*\circled[1]{\tikz[baseline=(char.base)]{
    \node[shape=circle,draw,inner sep=1pt, scale=0.4] (char) {#1};}}

\renewenvironment{quote}{%
   \list{}{%
     \leftmargin0.75cm   
     \rightmargin0.5cm
   }
   \item\relax
}
{\endlist}

\newenvironment{indquote}{
    \begin{quote}
    \fontfamily{lmtt}\selectfont\small
}
{
    \end{quote}
}

\newtheorem{theorem}{Theorem}

\newtheorem{definition}{Definition}
\newtheorem{lemma}{Lemma}
\newtheorem{proposition}{Proposition}

\theoremstyle{definition}
\newtheorem{remark}{Remark}
\newtheorem{example}{Example}

\def\cali{\mathcal{I}}

  \def\calt{\mathcal{T}}
\def\calr{\mathcal{R}} \def\calv{\mathcal{V}} \def\calm{\mathcal{M}}

\newcommand{\abs}[1]{\left| #1 \right|}

\def\cenv{C_{\circled{M}}}

\def\c2s{C^{2s}_{\circled{M}}}
\def\hc2s{\widehat C^{2s}_{\circled{M}}}

\begin{document}

\title[Affirmative Action in India]{Can Economic Theory Be Informative for the Judiciary?\\ Affirmative Action in India via\\
Vertical and Horizontal Reservations}


\author[S\"{o}nmez and Yenmez]{Tayfun S\"{o}nmez \and
M. Bumin Yenmez}

\thanks{Both S\"{o}nmez and Yenmez are affiliated with the Department of Economics, Boston College, 140 Commonwealth Ave, Chestnut Hill, MA, 02467. Emails: \texttt{sonmezt@bc.edu},
\texttt{bumin.yenmez@bc.edu}. S\"{o}nmez acknowledges the research support of Goldman Sachs Gives via
Dalinc Ariburnu---Goldman Sachs Faculty Research Fund. We thank Utku \"{U}nver,  participants at
Fall 2019 NBER Market Design Working Group Meeting, and especially five anonymous referees whose comments significantly improved
the exposition of the paper.}

\date{First version: March 2019,  this version: January 2021.
This version subsumes and replaces two distinct working papers  ``Affirmative Action in India via Vertical and Horizontal Reservations''
\citep{sonyen19} and ``Affirmative Action with Overlapping Reserves'' \citep{sonmez/yenmez:20}.}

\maketitle

\begin{abstract}
Sanctioned by its constitution, India is home to the world's most comprehensive affirmative action
program, where historically discriminated groups are protected with vertical reservations
implemented as ``set asides,'' and other disadvantaged groups are protected with
horizontal reservations implemented as ``minimum guarantees.''
A mechanism mandated by the Supreme Court in 1995 suffers from important anomalies,
triggering countless litigations in India.
Foretelling a recent reform correcting the flawed mechanism,
we propose the 2SMG mechanism that resolves all anomalies, and characterize it
with desiderata reflecting laws of India.
Subsequently rediscovered  with a high court judgment and enforced in Gujarat,
2SMG is also endorsed by \textit{Saurav Yadav v. State of UP (2020)\/},
in a Supreme Court ruling that rescinded the flawed mechanism.
While not explicitly enforced, 2SMG is indirectly enforced for an important subclass of applications in India,
because no other mechanism satisfies the new mandates of the Supreme Court.
\end{abstract}



\newpage

\section{Introduction} \label{sec:introduction}
Sanctioned by its constitution, India is home to one of the world's largest  affirmative action
programs. Allocation of government positions and seats at publicly funded educational institutions
are governed by the mandates outlined in the landmark Supreme Court judgment
\textit{Indra Sawhney  and others v. Union of India (1992)\/},\footnote{The case  is available at
\url{https://indiankanoon.org/doc/1363234/} (last accessed on 01/19/2021).}
widely known as the \textit{Mandal Commission Case\/}.
Under these mandates,
a mechanism that otherwise allocates positions based on an objective merit list of candidates is amended to implement
two types of affirmative action policies known as \textit{vertical reservations\/} (VR) and \textit{horizontal reservations\/} (HR).
Of the two policies, the VR policy is envisioned as a higher-level protection policy,
and, as such, it is mandated to be implemented on an ``over-and-above'' basis.
This means that if a member of a VR-protected class is ``entitled'' to an open position based on merit, then she must be
awarded an open position and not use up a VR-protected position.
This higher-level protection policy has been largely intended  for
historically oppressed classes, most notably
\textit{Scheduled Castes (SC), Scheduled Tribes (ST),\/} and \textit{Other Backward Classes (OBC)\/}.
The HR policy, on the other hand,  is envisioned as a lower-level protection policy,
and, as such, it is mandated to be implemented on a ``minimum guarantee'' basis.
This means that any position awarded to a member of an HR-protected group
counts toward HR protections.

\subsection{The Supreme Court judgment \textit{Anil Kumar Gupta (1995)\/} and Its Consequences} \label{intro-AKG}

In the absence of the HR policy, implementation of the VR policy is a straightforward task with a simple two-step procedure.
First, open positions are awarded to  individuals with the  highest merit rankings (including those from
VR-protected classes), and next, for each VR-protected class, positions  set  aside for this class are
awarded to members with the highest merit rankings  who have not yet received one of the open positions.
We refer to this procedure as the \textit{over-and-above choice rule.\/}
Most applications in the field, however,  also involve the HR policy,\footnote{For example,
\textit{persons with disabilities\/} are granted horizontal reservations at the federal level with the  Supreme Court  judgment
\textit{Union Of India \& Anr vs National Federation Of The Blind \& ... on 8 October, 2013\/}, available at 
\url{https://indiankanoon.org/doc/178530295/} (last accessed on 01/23/2021).
Moreover, in several states, women are also granted horizontal reservations by high court decisions.
Examples include Bihar with 35\% (of the positions), Andhra Pradesh and Gujarat with $33.3$\% each,  and Madhya Pradesh, Uttarakhand,
Chhattisgarh, Rajasthan, and Sikkim with 30\% each.}
and it is less clear how  the two policies can be implemented concurrently in this more elaborate version of the problem.
While the principles that guide the implementation
of reservation policies are clearly laid out in  \textit{Indra Sawhney (1992)\/} when the two protective policies are implemented independently, 
no guidance is provided for their concurrent implementation by the important judgement.
This gap has been later filled  in  \textit{Anil Kumar Gupta v. State of U.P. (1995)\/}, another judgment of the Supreme Court,
where an explicit procedure for the concurrent implementation of
VR and HR policies is devised and enforced in India.\footnote{The case
is available at \url{https://indiankanoon.org/doc/1055016/} (last accessed on 03/10/2019).}
For the past quarter century,  this judgment has served as a main reference for virtually all subsequent litigations on concurrent
implementation of VR and HR policies, of which there are thousands.
This is our starting point, where  our original motivations in writing this paper were;  (i) formulating an important flaw in the procedure
mandated under this judgment, (ii) documenting its adverse consequences in India, and (iii) advocating for an alternative procedure as a remedy.

The procedure enforced under \textit{Anil Kumar Gupta (1995)\/} first derives a tentative outcome using the
over-and-above choice rule, then it makes any necessary replacements for the tentative recipients of open positions
to accommodate HR protections within open positions, and finally it makes any necessary replacements for the tentative recipients of the VR-protected positions
to accommodate HR protections within VR-protected positions.
We refer to this procedure as the \textit{SCI-AKG choice rule\/}.
One critical mandate in this judgment, however, has introduced two related and highly consequential anomalies into the procedure,  often
generating unintuitive outcomes at odds with the philosophy of affirmative action, and
thereby sparking thousands of litigations  in India for the next 25 years. To present the scale of the resulting disarray,
some of the key litigations triggered by the flawed mandate are documented in detail in
Section \ref{sec:India} of the Online Appendix.\footnote{A simple search of the phrase ``horizontal reservation'' via Indian Kanoon, a free search engine for Indian Law,
reveals the scale of the litigations relating to this concept. Excluding cases at lower courts, as of 06/08/2021 there are 2067 cases at the Supreme Court and State High Courts
related to the implementation of horizontal reservations.}

The root cause of the failure of the SCI-AKG choice rule boils down to its exclusion of the members of VR-protected classes
from any replacements necessary to accommodate the HR protections for open positions.
Thus, members of the higher-privilege \textit{general category\/}, i.e.  individuals who are not members  of the VR-protected classes, are the  only ones entitled to
replace the tentative holders of the open positions to accommodate its HR protections.
This restriction regularly created situations in India where higher-merit individuals from VR-protected classes lose their positions to
lower-merit individuals from the higher-privilege general category,
an anomaly we refer to as a failure of \textit{no justified envy\/}.
The same flaw also created a conflict for individuals who qualify for both types of protections, since for these individuals
claiming their VR protections would mean giving up their HR protections for open positions, an anomaly we refer to as a failure of \textit{incentive compatibility\/}.
Both types of failures have been originally formulated in  \cite{aybo16} in the context of Brazilian college admissions,
although our paper is the first one to document their disruptive implications through
numerous litigations and interrupted recruitment processes.

Since the root cause of the crisis is the exclusive access given  to general-category individuals for HR protections within open positions,
a simple and intuitive solution lies in the removal of this restriction in the SCI-AKG  choice rule, thus making everyone eligible.
Focusing on applications with \textit{non-overlapping horizontal reservations\/}
where each individual qualifies for at most one HR protection,
we refer to this alternative choice rule as the \textit{two-step minimum guarantee (2SMG)\/} choice rule.
This version  of the problem is important both because it is widespread in the field with \textit{persons with disabilities\/} being the
only group granted with HR protections at the federal level, and also because the judgments on HR protections
abstract away from the technical aspects of overlapping horizontal reservations.

\subsection{A Resolution with the Supreme Court judgment \textit{Saurav Yadav (2020)\/}} \label{sec:yadav}
Prior to the March 2019 circulation of the first draft of our paper,  the above-presented failure of the  SCI-AKG choice rule
has never been directly addressed by the  highest court of India, despite the large scale disarray it created  in the country for 25 years.
As we have emphasized earlier, this was our primary motivation in writing this paper, along with
formulating the 2SMG choice rule as a possible replacement for the SCI-AKG choice rule.
However, a key December 2020 judgment by a three-judge bench of the Supreme Court,
not only changed this situation, but also reshaped some of the questions of interest for our paper,
while it was under  revision for this journal. The material presented in the rest of the Introduction reflects
the revisions to our paper following this important judgment.

Based on arguments parallel to our analysis and using several of the high court judgments presented in
Section \ref{sec:India} of the Online Appendix, the justices reached some of the same conclusions in
\textit{Saurav Yadav \& Ors v. State of Uttar Pradesh \& Ors (2020)\/}\footnote{judgment available at \url{https://indiankanoon.org/doc/27820739/},
last accessed 06/05/2021.}
as we have reached earlier in our paper.
For our purposes, the key aspects of this judgment can be summarized as follows:
\begin{enumerate}
\item The axiom of \textit{no justified envy\/} is mandated for all choice rules used in India.
\item  The SCI-AKG choice rule is rescinded due to its failure to satisfy \textit{no justified envy\/}.
\item As a possible replacement for the SCI-AKG choice rule, the 2SMG choice rule is endorsed,
although it is not explicitly mandated.\footnote{Prior to its endorsement by the Supreme Court, the 2SMG choice rule
was formulated in the August 2020 High Court of Gujarat judgment
\textit{Tamannaben Ashokbhai Desai v. Shital Amrutlal Nishar (2020\/)}, which mandated it for the state of Gujarat.
judgment available at \url{https://www.livelaw.in/pdf_upload/pdf_upload-380856.pdf}, last accessed 06/06/2021.
Our formulation and advocacy of the 2SMG choice rule also precedes this judgment, which we believe is the first
judgment in India to formulate this choice rule.}
\item Clarity is brought to which positions awarded to members of VR-protected classes are to be used up from open positions,
rather than the VR-protected positions.
\end{enumerate}
Apart from correcting a flawed mandate from \textit{Anil Kumar Gupta (1995)\/} with highly disruptive consequences,\footnote{The justices
of the Supreme Court indicate in  \textit{Saurav Yadav (2020)\/} that the flaw in the SCI-AKG choice rule is based on a
misinterpretation of \textit{Anil Kumar Gupta (1995)\/}.}
 \textit{Saurav Yadav (2020)\/} brings
clarity on the meaning of VR protections in the presence of HR protections, at a level that was never done before.
The defining characteristic of the VR protections is originally formulated in \textit{Indra Sawhney (1992)\/} as follows: \smallskip
\begin{indquote}
It may well happen that some members belonging to, say Scheduled Castes get selected in the open competition field on the basis of their own merit;
they will not be counted against the quota reserved for Scheduled Castes; they will be treated as open competition candidates.\smallskip
\end{indquote}
In the absence of the HR protections, the interpretation of this formulation is straightforward.
Up to the number of open positions,
individuals with the highest merit rankings are to be assigned the open positions,
whether they are members of a VR-protected class or not.
However, prior to \textit{Saurav Yadav (2020)\/},
a formal interpretation was never provided for the following question:
What does it mean to be \textit{selected in the open competition field on the basis of one's own merit\/} in the presence of HR protections?
For the case of non-overlapping horizontal reservations, this question is answered as follows in \textit{Saurav Yadav (2020)\/}:
Any individual who is entitled to an open position based on her merit score, \textit{including those who are entitled to one
to accommodate the HR protections\/}, is considered as an individual  who is selected
in the open competition field on the basis of her own merit.
The justices in \textit{Saurav Yadav (2020)\/} mandate that all such individuals are to be assigned open positions, thereby not
using up any VR-protected positions. We formulate this mandate in Section  \ref{subsec:implicit} as the axiom of \textit{compliance with VR protections\/}.

Altogether,  the axioms of \textit{no justified envy\/}, \textit{compliance with VR protections\/},
\textit{maximal accommodation of HR protections\/} (which means that all HR protections must be accommodated up to the number
of eligible applicants), and \textit{non-wastefulness\/} (which means that no position should remain idle while there are eligible applicants),
each mandated under  \textit{Saurav Yadav (2020)\/},
uniquely characterize the 2SMG choice rule for problems with non-overlapping horizontal reservations (Theorem \ref{thm:2smg}).
Therefore, even though  the 2SMG choice rule is merely endorsed but not enforced under \textit{Saurav Yadav (2020)\/},
the mandates in this judgment indirectly enforce it for field applications with non-overlapping horizontal reservations. 
Some of our main contributions in this revised paper include the characterization
of the 2SMG choice rule with axioms that directly formulate the mandates in \textit{Saurav Yadav (2020)\/} and what it means
for India, i.e. the observation that  the 2SMG choice rule in indirectly enforced with this judgment.

\subsection{Extended Analysis and Policy Advice for Overlapping Horizontal Reservations}

In contrast to field applications with non-overlapping horizontal reservations where  \textit{Saurav Yadav (2020)\/} has a very sharp
policy implication, as has its predecessors,
the judgment leaves some flexibility for applications with overlapping horizontal reservations.
In this version of the problem, an individual can benefit from multiple HR protections.
We make our most significant conceptual and theoretical  contributions for this general version of the problem.

Consider an individual who is a member of multiple groups, each of which is eligible for HR protections.
For example, a woman with a disability can benefit from HR protections both for \textit{women\/} and also for
\textit{persons with disabilities\/}.
The law does not specify whether this individual accommodates the minimum guarantees for all HR protections
she is qualified for, in this example  both for women and for persons with disabilities, or for only one of them.
We refer to the first convention as  \textit{one-to-all HR matching\/} and the second convention as  \textit{one-to-one HR matching\/}.
While the law is silent on this aspect of the problem, we advocate for the one-to-one HR matching convention for two reasons.
The first reason is technical: Adopting the alternative one-to-all HR matching convention introduces complementarities between individuals,
which in turn renders the problem computationally hard in general and allows for multiplicities.
In the above example, the admission of a man with no disability may depend on the admission of a woman with disability.
The second reason is practical: In many real-life applications in India, the number of positions are announced for
vertical category-horizontal trait pairs, which automatically embeds the one-to-one HR matching convention into the problem.

Under the one-to-one HR matching convention, an additional matching problem is essentially built into the original problem,
where a secondary task matches individuals to different types of HR protections to account for these protections.
Fortunately, this secondary task can be formulated as a \textit{maximum bipartite matching problem\/},
a well-studied problem in the combinatorial optimization literature.
Moreover, this  approach not only allows us to formulate  natural and immediate extensions of all
four axioms,  but also allows for a natural extension of the 2SMG choice rule in the \textit{two-step meritorious horizontal\/} (2SMH) choice rule.
In our main theoretical result (Theorem \ref{thm:2s-envchar}),
we extend our characterization of the 2SMG choice rule for the case of non-overlapping horizontal reservations
to the 2SMH choice rule  for the general case of  overlapping horizontal reservations.

\subsection{Organization of the Rest of the Paper}
After introducing the model
in Section  \ref{sec:model}, we present analysis and policy implications for problems with
non-overlapping horizontal reservations in Section \ref{sec:non-overlapping}.
The failure of the mechanism mandated by \textit{Anil Kumar Gupta (1995)\/},  its resolution by \textit{Saurav Yadav (2020)\/},
and the formal relation between these Supreme Court judgments and our analysis are also presented in this section.
Section \ref{sec:overlapping} presents an analysis of the model in its full generality with overlapping horizontal reservations, along with the
related theoretical literature.
We conclude with an epilogue in Section \ref{sec:Epilogue} and present all proofs in the Appendix.
Finally, we relegate the institutional background on VR and HR policies, extensive evidence from Indian court rulings
on the disruption caused by the SCI-AKG choice rule, and the equivalence of our formulation of the SCI-AKG choice rule with
its original formulation in  \textit{Anil Kumar Gupta (1995)\/}  to the Online Appendix.

\section{Model and Vertical/Horizontal Reservations}\label{sec:model}
There exists a finite set of individuals $\cali$ competing for $q$ identical positions.
Each individual  $i \in \cali$ is in need of a single position,  and has a distinct merit score $\sigma(i)\in \mathbb{R}_{+}$.\footnote{While students
can have the same merit score in practice,  tie-breaking rules are used to strictly rank them. For example, the Union
Public Service Commission uses age and exam scores to break ties.
See \url{https://www.upsc.gov.in/sites/default/files/TiePrinciplesEngl-26022020-R.pdf} (last accessed on 6/7/2020).}
While individuals with higher merit scores have higher claims for a position
in the absence of affirmative action policies, disadvantaged populations are protected
through two types of affirmative action policies,
(i) \textbf{vertical reservation (VR)} policies providing  ``higher level''  \textbf{VR protections},  and
(ii) \textbf{horizontal reservation (HR)} policies providing  ``lower level'' \textbf{HR protections}.

\subsection{Vertical Reservations}\label{sec:ver}
There exists a  set of \textbf{reserve-eligible categories} $\calr$ and a  \textbf{general category} $g \not\in \calr$.
Each individual belongs to a single category in $\calr \cup \{g\}$.
Define the (reserve-eligible) category membership function $\rho: \cali \rightarrow \calr \cup \{\emptyset\}$ such that,
for any individual $i \in \cali$,
\begin{itemize}
\item[] $\rho(i) = c$ \, indicates that $i$ is a member of the reserve-eligible category $c\in \calr$, and
\item[] $\rho(i) = \emptyset$ \, indicates that $i$ is a member of the general category $g$.
\end{itemize}

Given a set of individuals $I \subseteq \cali$ and a reserve-eligible category $c\in \calr$, define
\[ I^c = \{i \in I : \rho(i) = c\} \]
as the set of individuals in $I$ who are members of the reserve-eligible category $c\in \calr$.
Given a set of individuals $I \subseteq \cali$, define
\[ I^g = \{i \in I : \rho(i) = \emptyset\} \]
as the set of individuals in $I$ who are members of the general category $g$.

There are $q^c$ positions exclusively set aside  for the members of category $c\in \calr$.
We refer to these positions as \textbf{category-c positions}.
In contrast, members of the general category do not receive any special provisions under the VR policies.
Therefore,
\[ q^o=q-\sum_{c\in \calr} q^c
\]
positions are open for all individuals.  We refer to these positions as \textbf{open-category positions} (or \textbf{category-o positions}).
Let $\calv = \calr \cup \{o\}$ denote the set of \textbf{vertical categories for positions}.

It is important to emphasize that, in contrast to category-$c$ positions that are exclusively reserved for the
members of category $c\in \calr$, open-category positions are available for all, and hence they are not exclusively reserved for the members
of the general category $g$.
\begin{definition}
Given a reserve-eligible category $c \in \calr$,
an individual $i \in \cali$ is \textbf{eligible for category-}{\boldmath$c$} \textbf{positions} if,
\[ \rho(i) = c. \]
Any individual $i \in \cali$ is \textbf{eligible for open-category positions}.
\end{definition}
Given a category $v\in \calv$,
let $\cali^{v} \subseteq \cali$ denote the set of individuals who are eligible for category-$v$ positions.

VR protections have one important property that makes them the ``higher level'' affirmative action policy.
Positions that are earned by the members of reserve-eligible categories without invoking the VR protections, and thus
on the basis of their merit scores only, do not count against the VR-protected positions.
In this sense, VR protections are implemented on an ``over-and-above'' basis.

\subsection{Single-Category Choice Rule, Choice Rule, and Aggregate Choice Rule} \label{subsec-choicerule}
We next formulate the solution concepts used in our paper.

\begin{definition}
Given a category $v \in \calv$, a \textbf{single-category choice rule}
is a function $C^v: 2^{\cali} \rightarrow 2^{\cali^v}$ such that,  for any $I\subseteq \cali$,
\[ C^v(I)\subseteq I \cap \cali^v  \quad \mbox{ and  } \quad  \abs{C^v(I)}\leq q^v.\]
\end{definition}

\begin{definition}
A \textbf{choice rule} is a multidimensional function $C = (C^v)_{v \in \calv} :  2^{\cali} \rightarrow \prod_{v \in \calv}2^{\cali^v}$ such that,  for any $I\subseteq \cali$,
\begin{enumerate}
\item for any category $v \in \calv$,
\[ C^v(I)\subseteq I \cap \cali^v  \quad \mbox{ and  } \quad  \abs{C^v(I)}\leq q^v,
\]
\item for any two two distinct categories $v,v'\in \calv$,
\[ C^v(I)\cap C^{v'}(I)=\emptyset. \]
\end{enumerate}
\end{definition}

In addition to specifying the recipients, our formulation of a choice rule also specifies the categories of their positions.

\begin{definition}
For any choice rule $C=(C^v)_{v \in \calv}$,  the resulting \textbf{aggregate choice rule} $\widehat{C}: 2^{\cali} \rightarrow 2^{\cali}$ is given as
\[ \widehat{C}(I)=\bigcup_{v\in \calv} C^v(I) \qquad \mbox{for any } I\subseteq \cali.\]
\end{definition}
For any set of individuals, the aggregate choice rule yields the set of chosen individuals across all categories.

In the absence of horizontal reservations, which will be introduced in Section \ref{sec:hor},
the following three principles mandated in \textit{Indra Sawhney (1992)\/} uniquely define a choice rule,
thus making the implementation of VR policies straightforward.
First, an allocation must respect \textit{inter se\/} merit: Given two individuals from the same category,
if the lower merit-score individual is awarded a position, then the higher merit-score individual  must also be awarded a position.
Next, VR protections must be allocated on an ``over-and-above'' basis; i.e., positions that can be received without invoking the VR
protections do not count against VR-protected positions.
Finally, subject to eligibility requirements,  all positions have to be filled without contradicting the two principles above.
It is easy to see that these three principles uniquely imply the following choice rule:
First, individuals with the highest merit scores are assigned the open-category positions.
Next, positions reserved for the reserve-eligible categories are assigned to the remaining members of these categories, again based on their merit scores.
We refer to this choice rule as the \textit{over-and-above choice rule\/}.

\subsection{Horizontal Reservations within Vertical Categories}\label{sec:hor}
In addition to the reserve-eligible categories in $\calr$ that are associated with the higher level VR protections,
there is a finite set of traits $\calt$ associated with the lower level HR protections.
Each individual has a (possibly empty) subset of traits,  given by the trait function $\tau : \cali \rightarrow 2^{\calt}$.
Each trait represents a societal disadvantage, and
individuals who have this trait are provided with easier access to positions through a second type of affirmative action policy.

HR protections are provided within each vertical category.\footnote{This is not a federal mandate in India but rather a formal recommendation by
the Supreme Court judgment \textit{Anil Kumar Gupta (1995)\/}. The vast majority of the institutions in India follow this recommendation in
implementing HR policies in this form, also called \textit{compartmentalized horizontal reservations\/}.}
For any reserve-eligible category $c\in \calr$  and trait $t\in \calt$,
subject to the availability of qualified individuals,
a minimum of $q^{c}_t$  category-$c$ positions are to be assigned to  individuals from category $c$ with trait $t$.
We refer to these positions as \textbf{category-}{\boldmath $c$} \textbf{HR-protected positions for trait} {\boldmath$t$}.
Similarly, for any trait $t\in \calt$ and subject to the availability of individuals with trait $t$,
a minimum of $q^{o}_t$  open-category positions are to be assigned to  individuals with trait $t$.
We refer to these positions as \textbf{open-category HR-protected positions for trait} {\boldmath$t$}.

For each vertical category $v \in \calv$, we assume that the total number of category-$v$ HR-protected positions
is no more than the number of positions in category $v$. That is,  for each
category $v\in \calv$,
\[\sum_{t\in \calt} q^v_t \leq q^v.\]


We refer to HR policies where  an individual can have at most one trait as  \textbf{non-overlapping HR protections}, and
HR policies where  an individual can have multiple traits as  \textbf{overlapping HR protections}. 
In many field applications in India, HR protections are non-overlapping.\footnote{That is in part because \textit{persons with disabilities\/} are the only
group that is explicitly granted HR protections at the federal level.}
Unlike this version of the problem which is relatively less complex,
analysis of the  problem with overlapping HR protections
introduces a number of subtleties.

In contrast to VR protections, which are provided on an ``over-and-above'' basis,
HR protections are  provided within each vertical category on a ``minimum guarantee'' basis.
This means that positions obtained without invoking any HR protection still  accommodate the HR protections.\footnote{The official language used for
the distinction  between HR protections and VR protections is given in Section  \ref{sec:RKD2007} of the Online Appendix.}

Given a category $v \in \calv$ and assuming that HR policies are non-overlapping,
category-$v$  HR protections can be implemented with the following
(category-$v$) \textit{minimum guarantee choice rule\/} $C^v_{mg}$ \citep{echyen12}. \medskip

\begin{quote}
 \noindent{}{\bf Minimum Guarantee Choice Rule} {\boldmath$C^v_{mg}$}

\noindent Given a set of individuals $I \subseteq \cali^v$,\\
        \noindent{}{\bf Step 1:} for each trait $t\in \calt$,
    choose all individuals in $I$ with trait $t$  if the number of trait-$t$ individuals in $I$ is less than or equal to $q^v_t$, and
$q^v_t$ highest merit-score individuals in $I$ with trait $t$ otherwise.\\
		\noindent{}{\bf Step 2:} For positions unfilled in Step 1, choose unassigned individuals in $I$ with highest merit scores.\medskip
\end{quote}

The reason for restricting attention to problems with non-overlapping HR protections in defining this choice rule
is technical.  It is easy to see that the processing sequence of traits  in Step 1 of the procedure becomes immaterial for this case.
In contrast, the processing sequence of traits can affect the outcome under overlapping HR protections.
Moreover, in this more general case, and even with the additional specification of a trait processing sequence,
it is not clear whether the resulting choice rule is equally plausible for implementing the HR protections.
Indeed, in Section \ref{sec:smarthorizontal} we advocate for an alternative approach in extending the \textit{minimum guarantee choice rule\/}
for problems with overlapping HR protections.

\section{Analysis and Policy Implications with Non-Overlapping HR Protections} \label{sec:non-overlapping}

In this section, we present an analysis of concurrent implementation of VR and non-overlapping HR protections. 
Therefore, throughout this section, each individual is assumed to have at most one trait.
While Indian judgments and legislation on VR and HR policies more broadly apply to applications with 
overlapping HR protections as well, as we show in this section they have sharper implications
for field  applications with  non-overlapping HR protections. Moreover,  choice rules that have been either mandated or endorsed by
the Supreme Court since \textit{Indra Sawhney (1992)\/} all abstract away from any details pertaining to
overlapping HR protections.
Therefore, our analysis of this  more restrictive version of the model in this section has more direct policy implications in India.


\subsection{SCI-AKG Choice Rule and Its Flaws} \label{subsec-flaws}
We start our analysis
by introducing the \textit{SCI-AKG choice rule\/} that was mandated in India  for 25  years until December 2020.
The following definition simplifies the description of the  SCI-AKG choice rule.

\begin{definition}
A member of a reserve-eligible category $i \in \cup_{c\in \calr}\cali^{c}$ is a \textbf{meritorious reserved candidate}
if she  has one of the $q^o$ highest  merit scores among all individuals in $I$.
\end{definition}
Let $I^m$ denote the set of meritorious reserved candidates.

We are ready to formulate the \textit{SCI-AKG choice rule\/},  originally introduced in the Supreme Court judgment  \textit{Anil Kumar Gupta (1995)\/}
for the case of a single trait.\medskip

\begin{quote}
        \noindent{}{\bf SCI-AKG Choice Rule} {\boldmath$C^{SCI}$} = {\boldmath$(C^{SCI,\nu})_{\nu \in \calv}$}

\noindent Given a set of individuals $I \subseteq \cali$,
\begin{eqnarray*}
&&   C^{SCI,o}(I) = C^{o}_{mg}(I^m \cup I^g), \mbox{ and}\\
&&  C^{SCI,c}(I) = C^{c}_{mg}\big(I^c\setminus C^{o}_{mg}(I^m \cup I^g)\big) \quad \mbox{ for any } c \in \calr.
\end{eqnarray*}
       \end{quote}

It is important to emphasize that the formulation of the SCI-AKG choice rule given above is not the original formulation
presented in \textit{Anil Kumar Gupta (1995)\/}. The original formulation is based on first tentatively allocating the positions based on  the
over-and-above choice rule presented in Section \ref{sec:ver},
and subsequently carrying out any necessary adjustments to accommodate the HR protections.
We instead present a simpler formulation of the SCI-AKG choice rule,
using its relation to the \textit{minimum guarantee choice rule\/} introduced in Section \ref{sec:hor}.\footnote{The original description
of the SCI-AKG choice rule  in the Supreme Court judgments  \textit{Anil Kumar Gupta (1995)\/}
and \textit{Rajesh Kumar Daria (2007)\/}, and the result that shows the outcome equivalence of this formulation
can be seen in Section \ref{sec:AKG95} of the Online Appendix.}
It is also important to note that, while the justices formally introduced the SCI-AKG choice rule only for the case of a single trait,
their formulation immediately extends to multiple traits  assuming HR protections are non-overlapping.
Later in Section \ref{sec:singlecategory}, we show that extending the SCI-AKG choice rule to the more general version of problem 
with  overlapping HR protections introduces a number of subtleties, allowing for multiple generalizations
of this rule.

We next show that the SCI-AKG choice rule has two important flaws even for the simple case with a single trait.

\begin{example}\label{ex:ic}
There are VR protections for members of  a reserve-eligible category $c \in \calr$ and HR protections for women.
The set of individuals $I = \{m_1^g, m_2^g, w_1^g, m_1^c, w_1^c\}$ consists of two general category men
$m_1^g$, $m_2^g$, one general-category woman
$w_1^g$, one category-$c$  man  $m_1^{c}$, and  one category-$c$ woman $w_1^{c}$.
There are two open-category positions and one VR-protected position for category $c$.
One of the open-category positions is HR-protected for women.
Individuals have the following merit ranking:
\begin{center}
$\sigma(m_1^g)>\sigma(m_2^g)>\sigma(m_1^{c})>\sigma(w_1^{c})>\sigma(w_1^{g})$.
\end{center}
Since there are two open-category positions and neither of the two highest merit score individuals are from the reserve-eligible
category $c$,  the set of meritorious reserved candidates is $I^m = \emptyset$.
Therefore, the set of individuals under consideration for open positions is $I^m \cup I^g = I^g = \{m_1^g, m_2^g, w_1^g\}$.
Since $w_1^g$ is the only woman in the set $I^m \cup I^g$, she is awarded the
open-category HR-protected position for women despite having the lowest merit score.
Woman $w_1^{c}$ is not eligible for this position, although she would be had she not declared her
category membership for the reserve-eligible category $c$. The other open-category position is awarded to the highest merit score individual $m_1^g$.
Hence,  $C^{SCI,o}(I) = C^{o}_{mg}(I^m \cup I^g) = \{m_1^g, w_1^g\}$.

Since there is no category-$c$ HR-protected position for women, the highest merit score category-$c$ individual receives the only
category-$c$ position, and hence
$C^{SCI,c}(I) = \{m_1^c\}$.
Therefore, the set of individuals who are each awarded a position under the SCI-AKG choice rule is $\widehat{C}^{SCI}(I) =  \{m_1^g, w_1^g, m_1^c\}$.

There are two troubling aspects of this outcome. The first issue is that, even though the category-$c$ woman
$w_1^{c}$ has a higher merit score than the general category woman  $w_1^g$, the latter receives a position while the former does not.
That is, contrary to the philosophy of affirmative action, a lower merit score individual from the (unprotected) general category receives
a position at the expense of a higher merit score individual from a protected category.
The second issue is that, since she is the highest merit score woman among all applicants,
 woman $w_1^{c}$ can receive the open-category HR-protected position for women simply by
not declaring her eligibility for the VR-protected position for category-$c$.
\qed
\end{example}

The shortcomings of the SCI-AKG choice rule presented in Example \ref{ex:ic} are not merely abstract possibilities,
but rather are highly visible flaws that have been responsible for thousands of  litigations that distrupt
recruitment processes throughout India, as documented in Section \ref{sec:challenges} of the Online Appendix.
The root cause of both anomalies is the restriction of the open-category HR protections to general
category individuals only.
This restriction creates an immediate  (and rather obvious) conflict for individuals who qualify  for both 
VR and HR protections: With the exception of meritorious reserved candidates,
any such individual loses her qualifications for open-category HR protections by claiming her VR protections.
Consequently, this conflict reflects itself in the following two deficiencies that go
against the philosophy of affirmative action:
\begin{enumerate}
\item  \textit{Possibility of a higher-merit protected individual losing a position to a lower-merit unprotected individual\/}:
For example, a woman from the VR-protected category Scheduled Castes may remain
unassigned while a lower merit-score woman from the higher-privilege general category  receives a position through open-category HR protections for women.
\item \textit{Necessity to give up VR protections to claim open-category HR protections\/}: For example, a woman from Scheduled Castes  may remain unassigned
by declaring her membership for Scheduled Castes, but she can  receive an open-category HR-protected position for women
by withholding her Scheduled Castes membership, and thus she benefits from not declaring this information.
\end{enumerate}
These deficiencies motivate our axioms of \textit{no justified envy\/} and \textit{incentive compatibility\/}.

The following \textit{HR-maximality function\/} plays a key role not only in our formulation of the axiom of \textit{no justified envy\/}, but also
in our formulation of two additional axioms introduced later in this section.
Moreover,  the extension of our analysis to the more general model with overlapping HR protections later presented in Section \ref{sec:overlapping}
also critically depends on the extension of this function.
\smallskip

\begin{definition} \label{def-HRmaximality}
Given a vertical category  $v \in \calv$, the (category-$v$) \textbf{HR-maximality function} {\boldmath $n^v : 2^{\cali^v} \rightarrow \mathbb{N}$} is defined as,
for any $I \subseteq \cali^v$,
\[ n^v(I) = \sum_{t \in \calt} \min\Big\{\big|\{i\in I \; : \; t\in\tau(i)\}\big| , \; q^v_t\Big\}.
\]
\end{definition}
Observe that, for any set of individuals $I$ who are eligible for category-$v$ positions, the  category-$v$ HR-maximality  function $n^v$ gives the maximum number of
category-$v$ HR-protected positions that can be awarded.\footnote{One way this maximum can be obtained is via the
category-$v$ minimum guarantee choice rule $C^v_{mg}$.}

\begin{definition}
A choice rule $C=(C^{\nu})_{\nu \in \calv}$ satisfies \textbf{no justified envy} if, for every $I\subseteq \cali$,\; $v \in \calv$,\;
$i\in C^v(I)$, and
$j \in \big(I\cap \cali^v\big) \setminus \widehat{C}(I)$, 
\[ \sigma(j) \, > \, \sigma(i)    \; \implies \;  n^v\Big(\big(C^v(I)\setminus \{i\}\big) \cup \{j\}\Big) < n^v(C^v(I)).
\]
\end{definition}
This axiom requires that, given two individuals who are both eligible for a position in a category,
the lower merit-score individual can receive a position at the expense of the higher merit-score individual
only if not doing so strictly decreases the number of HR protections that are accommodated in that category.
Therefore, under this axiom, increasing the utilization of HR protections in a category can be the only reason to award
a position at this category to a lower merit-score individual at the expense of an unassigned higher merit-score eligible individual.

We next formulate the axiom of  \textit{incentive compatibility\/}, first introduced by  \cite{aybo16} in their analysis
of the affirmative action policies in Brazilian college admissions:

\begin{definition}
An individual \textbf{withholds some of her reserve-eligible privileges} if she does not declare either
her reserve-eligible category membership or some of her traits.
\end{definition}
In India, individuals are not required to declare their reserve-eligible privileges.

\begin{definition}
A choice rule $C$ is \textbf{incentive compatible} if, for every $I \subseteq \cali$,  any individual $i\in I$ who is selected from $I$ under the aggregate choice rule
$\widehat{C}$ by withholding  some of her reserve-eligible privileges is also selected from $I$ under 
$\widehat{C}$  by declaring all her reserve-eligible privileges.
\end{definition}
Under a choice rule that satisfies this axiom,
privileges that are meant to provide positive discrimination would never produce the opposite effect and thus hurt an individual upon declaring eligibility.
Failure of incentive compatibility is implausible both from a normative perspective, since it is against the philosophy of affirmative
action, and also from a strategic perspective, since it may force individuals to withhold their privileges.
As we document clear evidence in Section \ref{sec:wrongful} of the Online Appendix, it also creates one additional difficulty in India.

Eligibility for VR protections typically depends on an  individual's caste membership.
While this information is supposed to be private information, it can often be inferred by the central planner due to various indications
such as the individual's last name. A central planner can also obtain this information through documents such as a diploma.
Hence, eligibility for VR protections may not be truly private information, and the lack of incentive compatibility of a choice rule
may enable a malicious central planner to exploit this information to deny an applicant her open-category HR protections.
As documented in Section \ref{sec:wrongful} of the Online Appendix, this type of misconduct not only has been widespread in parts of India, but it even appears to
be centrally organized by the local governing bodies in some of its jurisdictions.

\subsection{An Easy Fix:  2SMG Choice Rule} \label{subsec:2SMG}

Apart from its simplicity,
an additional advantage of formulating the SCI-AKG choice rule using its relation to the \textit{minimum guarantee choice rule\/} is that,
unlike its original formulation that obscures a possible remedy,
our equivalent formulation suggests an easy fix.
Both anomalies of the SCI-AKG choice rule are caused by the exclusive access given to the general-category individuals
for open-category HR protections.
This restriction reflects itself in our formulation of the SCI-AKG choice rule during  the derivation of  the open-category assignments
through the formula
\[ C^{SCI,o}(I) = C^{o}_{mg}(I^m \cup I^g).
\]
Observe that, instead of running the choice rule  $C^{o}_{mg}$ for the set of individuals $I^m \cup I^g$,
running it for the set of all individuals $I$ provides us with an immediate and intuitive fix. We refer to this alternative mechanism as the
\textit{two-step minimum guarantee (2SMG) choice rule\/}.

	\begin{quote}
        \noindent{}{\bf Two-Step Minimum Guarantee (2SMG) Choice Rule} {\boldmath$C^{2s}_{mg} = (C^{2s,\nu}_{mg})_{\nu \in \calv}$}\smallskip

\noindent Given a set of individuals  $I \subseteq \cali$,
\begin{eqnarray*}
&& C^{2s,o}_{mg}(I) = C^{o}_{mg}(I), \mbox{ and}\\
&& C^{2s,c}_{mg}(I) = C^{c}_{mg}\big(I^c\setminus C^{o}_{mg}(I)\big) \quad \mbox{ for any } c \in \calr.
\end{eqnarray*}
\end{quote}

Since the SCI-AKG choice rule is formally introduced in \textit{Anil Kumar Gupta (1995)\/} for the case of a single trait,
and in particular when HR protections are non-overlapping, 
it is best to consider the 2SMG choice rule  for the model with non-overlapping HR protections only.\footnote{When HR protections are overlapping, 
the outcome of the 2SMG choice rule depends on the processing sequence of traits at each vertical category of  positions.}

As one would naturally expect, replacing the SCI-AKG choice rule with the 2SMG choice rule results in a
weakly less favorable outcome for members of the general category.
The comparison for members of reserve-eligible categories is less straightforward, because in addition to the VR-protected
positions, these individuals  also compete for the open positions. However, assuming sufficient demand at each reserve-eligible category,
replacing the SCI-AKG choice rule with the 2SMG choice rule results in a
weakly more favorable outcome in aggregate for members of the reserve-eligible categories.

\begin{proposition} \label{prop:SCI-vs-2smg}
For every $I \subseteq \cali$,
\[ \widehat{C}^{2s}_{mg}(I) \cap I^g \subseteq \widehat{C}^{SCI}(I) \cap I^g,\]
and assuming $|I^c| \geq q^o + q^c$ for each reserve-eligible category $c \in \calr$,
\[  \sum_{c\in\calr } \big|\widehat{C}^{2s}_{mg}(I) \cap I^c\big| \geq \sum_{c\in\calr } \big|\widehat{C}^{SCI}(I) \cap I^c\big|.
\]
\end{proposition}

\subsection{The Demise of the SCI-AKG Choice Rule and the Rise of the 2SMG Choice Rule} \label{subsec-demiseandrise}

In a rather unexpected development and while this paper was under revision for this journal,
in \textit{Saurav Yadav (2020)\/} a three-judge bench of the Supreme Court declared that
the SCI-AKG choice rule  is a product of misinterpretation of the Court's earlier judgments.
Referring to the failure of the SCI-AKG choice rule to  satisfy \textit{no justified envy\/} as an ``incongruity,'' the
justices annulled this mechanism, since it can result in ``irrational'' results.
Importantly, the same judgment
also endorsed the \textit{2SMG choice rule\/} as a possible replacement  for the abandoned SCI-AKG choice rule.
While  the justices have not mandated the  \textit{2SMG choice rule\/} in  \textit{Saurav Yadav (2020)\/},
they mandated that any choice rule adopted in India satisfy the axiom of \textit{no justified envy\/} and further
brought clarity for one  additional  subtle aspect of the HR protections presented in Section \ref{subsec:implicit}.
Importantly,  the  \textit{2SMG choice rule\/} is the only mechanism that satisfies these new mandates
together with those from \textit{Indra Sawhney (1992)\/} in applications with non-overlapping HR protections.
We next present this significant implication of \textit{Saurav Yadav (2020)\/}, which is not observed
in this important judgment.

\subsection{The Implicit Mandate of  the 2SMG Choice Rule Under  \textit{Saurav Yadav (2020)\/}} \label{subsec:implicit}

We next formulate three additional axioms,  the first of which is originally mandated by \textit{Indra Sawhney (1992)} and
maintained by  \textit{Saurav Yadav (2020)\/}, whereas the latter two are only recently  mandated by \textit{Saurav Yadav (2020)\/}
(as in the case of  the \textit{no justified envy\/} axiom formulated in Section \ref{subsec-flaws})
at their strength formulated below.

\begin{definition}
A choice rule $C=(C^{\nu})_{\nu \in \calv}$ is \textbf{non-wasteful} if,
for  every $I\subseteq \cali$,\; $v \in \calv$, and $j \in I$,
\[j \not\in  \widehat{C}(I) \; \mbox{ and } \;  |C^v(I)| < q^v   \quad \implies \quad    j \not\in \cali^v.
\]
\end{definition}
That is, if an individual $j$ is declined a position from each one of the categories (thus remaining unmatched) while
there is an idle position at some category $v\in \calv$,
then it must be the case that individual $j$ is not eligible for a position at category $v$.
This mild efficiency axiom has been mandated in India since \textit{Indra Sawhney (1992)\/}.

\begin{definition}
A choice rule $C=(C^{\nu})_{\nu \in \calv}$ \textbf{maximally accommodates HR protections}, if
for every $I \subseteq \cali$,  $v \in \calv$,\; and
$j \in \big(I\cap \cali^v\big) \setminus \widehat{C}(I)$,
\[ n^v(C^v(I)) =  n^v(C^v(I) \cup \{j\}).
\]
\end{definition}
In words, an individual who remains unassigned should not be able to increase the utilization of
HR protections at any category where  she has eligibility, if she were to be instead assigned a position in this category.
The only reason this axiom was not mandated in India prior to \textit{Saurav Yadav (2020)\/}  is that
under the previous interpretation of \textit{Anil Kumar Gupta (1995)\/} members of reserve-eligible categories were considered ineligible
for open-category HR protections. This restriction, which has been the root cause of the controversies involving the
SCI-AKG choice rule, has been revoked by  \textit{Saurav Yadav (2020)\/}, and consequently the axiom of \textit{maximum
accommodation of HR protections\/} is mandated in its stronger form as formulated above.

\begin{definition} \label{def-VR}
A choice rule $C=(C^{\nu})_{\nu \in \calv}$  \textbf{complies with VR protections} if,
for  every $I\subseteq \cali$,\; $c \in \calr$, and $i\in C^c(I)$,
\begin{enumerate}
\item $|C^o(I)| = q^o$,
\item for every $j \in C^o(I)$,
\[ \sigma(j) < \sigma(i) \quad \implies \quad n^o\big(C^o(I)\big) > n^o\big((C^o(I)\setminus\{j\}) \cup \{i\}\big), \mbox{ and} \]
\item $n^o\big(C^o(I) \cup \{i\}\big) = n^o\big(C^o(I)\big).$
\end{enumerate}
\end{definition}
Here the first two conditions formulate the idea of a vertical reservation \`{a} la \text{Indra Sawhney (1992)\/}, and they are
directly implied by the concept of ``over-and-above.''
For an individual $i$ to receive a position set aside for a reserve-eligible category (thereby not receiving an open position),
it must be the case that each open position is either assigned to a higher merit-score individual $j$, or to an individual $j$ whose selection
instead of $i$ increases the utilization of open-category HR protections.
The third condition additionally requires that a member of a reserve-eligible category who can improve the utilization of  open-category
HR protections shall not use up a VR-protected position.
Importantly, this third condition is an implication of another mandate in \textit{Saurav Yadav (2020)\/}, and therefore this judgment
enforces the axiom  of \textit{compliance with VR protections\/} in its stronger form as formulated above.\footnote{See Section \ref{app-yadav}  of
the Online Appendix for this important mandate in  \textit{Saurav Yadav (2020)\/}.}

We are ready to present our first main result, one that has important and previously unknown policy implications for India.

\begin{theorem}\label{thm:2smg}
Suppose each individual has at most one trait.
A choice rule
\begin{enumerate}
\item maximally accommodates HR protections,
\item satisfies no justified envy,
\item is non-wasteful, and
\item complies with VR protections
\end{enumerate}
if, and only if, it is the 2SMG choice rule $C^{2s}_{mg}$.
\end{theorem}

Prior to its endorsement by the three-judge bench of the Supreme Court in  \textit{Saurav Yadav (2020)\/}, the 2SMG choice rule  has been
introduced by the justices of the High Court of Gujarat in \textit{Tamannaben Ashokbhai Desai (2020\/)}, an August 2020 judgment
which also mandated the 2SMG choice rule in the state of Gujarat.\footnote{Our introduction and advocacy of  the 2SMG choice rule
predates both of these important judgments.}
However, while this choice rule is merely endorsed and not explicitly mandated by   \textit{Saurav Yadav (2020)\/} throughout India,
our first main result in Theorem \ref{thm:2smg} implies that this important judgment has implicitly mandated this mechanism
in field applications with non-overlapping HR protections.

\section{General Analysis and Policy Recommendations with Overlapping HR Protections} \label{sec:overlapping}

To the best of our knowledge,  the judgments on the implementation of HR policies in India largely abstract away
from any technical complications due to  overlapping HR protections.
Since this more general version of the problem is fairly common in the field, in this section we extend  our analysis to the model with overlapping HR protections.
This version of the problem, however, introduces a subtle but critical technical consideration that allows for at least two approaches
to generalize our model.
Hence, before presenting an analysis of concurrent implementation of VR and overlapping HR protections, 
we first elaborate on this consideration and justify the  modeling choice we make for our generalization.

\subsection{One-to One vs One-to-All HR Matching}
Whether horizontal reservations are overlapping or not,
an individual loses her open-category HR protections upon declaring her VR protections
under the Supreme Court judgment \textit{Anil Kumar Gupta (1995)\/}. Therefore, the main flaws of the SCI-AKG choice rule, originally defined for a single trait,
carry over to any possible generalization with overlapping HR protections.
For this more general and complex case, however,
one technical and subtle aspect of implementation of HR protections has been left unlegislated
and remains at the discretion of the central planner.
The law is silent on whether the admission of an individual with multiple traits
accommodates the minimum guarantee requirements for all her traits or only for one of her traits.
For example, suppose there is one HR-protected position for women and one HR-protected position for persons with disabilities.
If a woman with a disability  is admitted, the law does not specify whether she is to accommodate the minimum guarantee requirements
both for women and also for persons with disabilities, or only for one of these two protected groups.

In our extension, we focus on the second convention of implementing the HR protections, and thus assume that
an individual counts toward the  minimum guarantee requirement for only one of her traits  upon admission.
We refer to this convention of implementing HR protections as \textbf{one-to-one HR matching}, and
the alternative convention (where an individual counts toward the  minimum guarantee requirements for all her traits upon admission) as \textbf{one-to-all HR matching}.
There are two reasons for this important modeling choice.

The first reason is technical. The alternative convention of one-to-all HR matching introduces complementarities between
individuals, making their admissions potentially contingent on each other.
For example, if there is one HR-protected position for women and one HR-protected position for persons with disabilities,
the admission of a man without a disability may depend on the admission
of a woman with a disability who can accommodate the HR protections for both protected groups.
This complementarity, in turn, not only renders the derivation of feasible groups of individuals computationally hard,
but it also makes any possible solution technically less elegant.\footnote{See Section 4 in \cite{sonmez/yenmez:20} for an analysis
under the one-to-all HR matching convention with two traits.}
In contrast, our adopted convention of one-to-one HR matching enables
a fairly clean and computationally simple solution, as we present later in this section.

The second reason is practical.
While either convention appears to be allowed by the Indian judgments and legislation,
we have been unable to find any application with overlapping HR protections
where the allocation rules clearly specify (or imply) the adoption of the one-to-all HR matching convention.
In contrast, in many field applications, the central planner announces the number of positions for each
category-trait pair,\footnote{See for example Table 2 in \textit{Saurav Yadav (2020)\/}.}
which implicitly implies that they adopt the one-to-one HR matching convention.\footnote{We are also able to find a field application,
where the allocation rules explicitly specify the adoption of the one-to-one HR matching convention.}


\subsection{Single-Category Analysis with Overlapping HR Protections} \label{sec:singlecategory}

Since HR policies are implemented within vertical categories, we  start our analysis 
with the simple case of a single category. This version of the problem  also relates to practical applications other than our
main application in India, such as the allocation of K-12 public school seats in Chile where there are overlapping HR protections  \citep{correa19}.

Throughout Section \ref{sec:singlecategory}, we fix a category $v \in \calv$.

\subsubsection{The Case Against a Fixed Processing Sequence of Traits}


The 2SMG choice rule, introduced in Section \ref{subsec:2SMG}, is not well-defined in problems with overlapping HR
protections, because, for any vertical category $v\in \calv$,
the outcome of the category-$v$ minimum guarantee choice rule may depend  on the processing sequence of traits.
Therefore, it may be compelling to resolve this multiplicity by simply specifying
a processing sequence of traits  for each vertical category as additional list of parameters of the choice rule.
However,  we caution against this (admittedly compelling) generalization for it may introduce additional flaws in the system.

\begin{example} \label{Ex1}
There is one category (say open category), three individuals $i_1,i_2,i_3$, and two positions. There are
two traits $t_1, t_2$, with one HR-protected position each.
Individual $i_1$ has both traits, individual $i_2$ has no trait, and individual $i_3$ has trait $t_1$ only. Individuals
are  merit ranked as
\[ \sigma(i_1) \; > \; \sigma(i_2) \; > \; \sigma(i_3).
\]
We next generate the outcome of the (open category) \textit{minimum guarantee\/} choice rule for both processing sequences of the two traits,  first by
processing the trait-$t_1$ minimum guarantee prior to the trait-$t_2$ minimum guarantee, and subsequently
by processing them in the reverse order. \smallskip

\noindent \textit{Trait $t_1$ first, trait $t_2$ next\/}: The highest merit score individual with trait $t_1$ is $i_1$; she
receives a position, accommodating the minimum guarantee for trait $t_1$. No remaining individual has trait $t_2$; therefore only individual $i_1$ receives
a position in Step 1. The highest merit score remaining individual $i_2$ receives the second position in Step 2.  The set of selected individuals
is $\{i_1,i_2\}$, and only the trait $t_1$ minimum guarantee is accommodated under the first processing  sequence of traits. \smallskip

\noindent \textit{Trait $t_2$ first, trait $t_1$ next\/}: The highest merit score individual with trait $t_2$ is $i_1$; she
receives a position, accommodating the minimum guarantee for trait $t_2$. Among the remaining individuals, the highest merit score individual with trait $t_1$ is $i_3$; she
receives a position, accommodating the minimum guarantee for trait $t_1$. No position remains, and therefore the set of selected individuals
is $\{i_1,i_3\}$. Minimum guarantees for both traits are accommodated  under the second processing  sequence of traits.
\qed
\end{example}
Example \ref{Ex1} shows that;
\begin{enumerate}
\item the outcome of the minimum guarantee choice rule, in general, depends on the processing sequence of traits, and
\item for some processing sequences of traits, it may accommodate fewer than the maximum possible HR protections.
\end{enumerate}
Essentially, Example \ref{Ex1}  shows that a fixed processing sequence of traits may result in denial of
HR protections which can be avoided.

Our next example reveals another problematic implication of implementing the minimum guarantee choice rule under
a fixed processing sequence of traits.

\begin{example} \label{Ex2}
There is one category (say the open category), four individuals $i_1,i_2,i_3,i_4$,  and three positions. There are
two traits $t_1, t_2$, with one HR-protected position each.
Individual $i_1$ has both traits, individual $i_2$ has no trait, individual $i_3$ has only trait $t_1$, and individual $i_4$ has only trait $t_2$.
Individuals are  merit ranked as
\[ \sigma(i_1) \; > \; \sigma(i_2) \; > \; \sigma(i_3) \; > \; \sigma(i_4).
\]
We next generate the outcome of the (open category) \textit{minimum guarantee\/} choice rule for both processing sequences of the two traits,  first by
processing the trait-$t_1$ minimum guarantee prior to the trait-$t_2$ minimum guarantee, and subsequently
by processing them in the reverse order. \smallskip

\noindent \textit{Trait $t_1$ first, trait $t_2$ next\/}: The highest merit score individual with trait $t_1$ is $i_1$; she
receives a position, accommodating the minimum guarantee for trait $t_1$.
Among the remaining individuals, the one with the highest merit score with trait $t_2$ is $i_4$; she
receives a position, accommodating the minimum guarantee for trait $t_2$ and finalizing Step 1. The last position is assigned in Step 2 to the
highest merit score remaining individual $i_2$, and therefore  the set of selected individuals
is $\{i_1,i_2,i_4\}$  under the first processing  sequence of traits.  \smallskip

\noindent \textit{Trait $t_2$ first, trait $t_1$ next\/}: The highest merit score individual with trait $t_2$ is $i_1$; she
receives a position, accommodating the minimum guarantee for trait $t_2$.
Among the remaining individuals, the one with the highest merit score with trait $t_1$ is $i_3$; she
receives a position, accommodating the minimum guarantee for trait $t_1$ and finalizing Step 1. The last position is assigned in Step 2 to the
highest merit score remaining individual $i_2$, and therefore  the set of selected individuals
is $\{i_1,i_2,i_3\}$  under the second processing  sequence of traits.
\qed
\end{example}
Example \ref{Ex2} reveals  that, depending on the processing
sequence of traits, the outcome of the minimum guarantee choice rule  may admit lower merit score individuals at the expense
of higher merit score ones without  affecting adherence to the horizontal reservation policies.
In Example \ref{Ex2},   when the merit based outcome of  $\{i_1,i_2,i_3\}$ already accommodates the HR protections,
there is clearly no reason to select a less meritorious group.

These two examples not only  guide us on adjustments of our axioms to account for  overlapping HR protections,
they also motivate the \textit{meritorious horizontal choice rule\/}, introduced in Section \ref{sec:smarthorizontal}, as a natural
extension of the 2SMG choice rule.

\subsubsection{HR Graph and the Generalized HR-maximality Function} \label{sec:scaxioms}

In contrast to the version of our model with non-overlapping HR protections where
maximizing the accommodation of HR protections is a straightforward task, doing the
same for the general version of the model with overlapping HR protections
requires  embedding a \textit{maximum trait matching\/} procedure within each category.
Therefore, we rely on the following construction to generalize our HR-maximality function,  which we will use
\begin{enumerate}
\item  to extend our axioms initially presented in Section \ref{sec:non-overlapping} for the model with non-overlapping HR protections, and
\item to generalize the 2SMG choice rule for the model with 
overlapping HR protections in a way that escapes the shortcomings presented in Examples \ref{Ex1} and \ref{Ex2}.
\end{enumerate}

Given a category $v \in \calv$ and a set of individuals $I \subseteq \cali^v$,
construct the following two-sided \textbf{category-}{\boldmath$v$} \textbf{HR graph}. On one side of the
graph, there are individuals in $I$. On the other side, there are HR-protected positions for category $v$.
Let $H_t^v$ denote the set of trait-$t$ HR-protected  positions for category $v$ and let $H^v = \bigcup_{t \in \calt} H^v_t$.
There are $q^v_t$ positions in $H^v_t$ and $\sum_{t\in \calt} q^v_t$ positions in $H^v$.
An individual $i \in I$ and a position $p \in H^v_t$ are \textbf{connected} in this graph if  and only if
individual $i$ has trait $t$.
\begin{definition}
Given a category $v \in \calv$ and a set of individuals $I \subseteq \cali^v$,
a \textbf{trait-matching} of individuals in $I$ with HR-protected  positions in $H^v$  is a
function $\mu: I \rightarrow H^v \cup \{\emptyset\}$ such that,
\begin{enumerate}
\item for any $i \in I$, $t\in \calt$,
\[ \mu(i)\in H_t^v \; \implies \; t \in \tau(i),  \]
\item for any $i,j \in I$,
\[ \mu(i) = \mu(j) \not= \emptyset \; \implies \; i=j.\]
\end{enumerate}
\end{definition}

\begin{definition}
Given a category $v \in \calv$ and a set of individuals $I \subseteq \cali^v$,
a trait-matching of individuals in $I$  with HR-protected  positions in $H^v$
\textbf{has maximum cardinality in a (category-}{\boldmath$v$}) \textbf{HR graph} if
there exists no other trait-matching that assigns a strictly higher number of HR-protected  positions to individuals.
\end{definition}
Let {\boldmath$n^v(I)$} denote the maximum number of category-$v$ HR-protected  positions that can
be assigned to individuals in $I$.\footnote{This number can be found through several polynomial time algorithms such as
\textit{Edmonds' Blossom Algorithm} \citep{edmonds_1965}.}
Observe that function $n^v$ generalizes the category-$v$ HR-maximality function presented in Definition \ref{def-HRmaximality}
for the model with non-overlapping HR protections to the more general version of the model with overlapping HR protections
(under the convention of one-to-one HR matching).

\begin{remark}
All our axioms in Section \ref{sec:non-overlapping} are extended for our more general model with overlapping HR protections
by simply replacing the simpler version of the HR-maximality function given in Definition \ref{def-HRmaximality} with
the generalized version.
\end{remark}

The following terminology is useful for our generalization of the 2SMG choice rule.

\begin{definition}
Given a category $v \in \calv$ and a set of individuals $I \subseteq \cali^v$,
an individual $i\in  \cali^v \setminus I$ \textbf{increases the (category-}{\boldmath$v$}) \textbf{HR utilization of} {\boldmath$I$} if
\[n^v(I \cup\{i\})=n^v\big(I)+1.\]
\end{definition}

\subsubsection{Meritorious Horizontal Choice Rule} \label{sec:smarthorizontal}
We are ready to introduce a single-category choice rule that escapes the shortcomings presented
in Examples \ref{Ex1} and \ref{Ex2}. The main innovation in this choice rule is the optimization it carries
out to determine who is to account for each minimum guarantee when some of the individuals can account for
one or another due to multiple traits they have.
Intuitively, this choice rule  exploits the flexibility in trait-matching in order to accommodate the
HR protections with higher merit-score individuals.

Given a category $v \in \calv$ and a set of individuals $I \subseteq \cali^v$, the outcome of this choice rule
is obtained using the following procedure.
\medskip

	\begin{quote}
        \noindent{}{\bf Meritorious Horizontal Choice Rule} {\boldmath$C^{v}_{\circled{M}}$}\smallskip

		\noindent{}{\bf Step 1.1}:
			Choose the highest merit-score individual in $I$ with a trait for an  HR-protected  position.
			Denote this individual by $i_1$
            and let $I_1 = \{i_1\}$. If no
            such individual exists, proceed to Step 2.
	
		\noindent{}{\bf Step 1.k} {\boldmath($k\in \{2, \ldots, \sum_{t\in \calt} q^v_t\}$)}:
		Assuming such an individual exists,
		choose the highest merit-score individual in $I \setminus I_{k-1}$
		who increases the HR utilization of $I_{k-1}$.\footnote{This can be done
           with various computationally efficient algorithms; see, for example, the bipartite cardinality matching algorithm \citep[Page 195]{lawler}.}
		Denote this individual by $i_k$ and  let $I_k=I_{k-1}\cup \{i_k\}$ .
		If no such individual exists, proceed to Step 2.

        \noindent{}{\bf Step 2}:
           For unfilled positions, choose unassigned individuals with highest merit scores until either all positions
           are filled or  all individuals are selected.
	\end{quote}
When the number of individuals is less than $q^v$, this procedure selects all individuals.
Otherwise, if there are more than $q^v$ individuals, then it chooses a set with $q^v$
individuals.

\subsubsection{Single-Category Results with Overlapping HR Protections}

We next present two single-category results under overlapping HR protections,
which suggest that the case for the meritorious horizontal choice rule is especially
strong in this framework.

Justifying the naming of this choice rule, our next result shows that the meritorious horizontal choice rule $C^v_{\circled{M}}$
always selects higher merit-score
individuals compared to other choice rules that maximally accommodate HR protections.

\begin{proposition}\label{prop:compenv}
Given a category $v \in \calv$,
let $C^v$ be any single-category choice rule that maximally accommodates HR protections.
Then, for every set of individuals $I \subseteq \cali^v$,
\begin{enumerate}
  \item $|C^v(I)| \leq |C^v_{\circled{M}}(I)|$, and
  \item for every $k\leq |C^v(I)|$, if $i$ is the $k$-th highest merit-score individual in $C^v_{\circled{M}}(I)$ and
  $j$ is  the $k$-th highest merit-score individual  in $C^v(I)$, then
\[
 i=j \quad \mbox{ or } \quad \sigma(i)  >  \sigma(j).
\]
\end{enumerate}
\end{proposition}

We next present a characterization of the meritorious horizontal  choice rule  $C^v_{\circled{M}}$.

\begin{theorem}\label{thm:envchar}
Given a category $v \in \calv$,
a single-category choice rule
\begin{enumerate}
\item maximally accommodates HR protections,
\item satisfies no justified envy, and
\item is non-wasteful
\end{enumerate}
if, and only if, it is the meritorious horizontal choice rule $C^v_{\circled{M}}$.
\end{theorem}

\subsection{Two-Step Meritorious Horizontal Choice Rule \& Its Characterization} \label{sec:multicategory}

We are ready to formulate and propose a choice rule for our model in its full generality. The following choice rule
uses the meritorious horizontal choice rule multiple times,  first to allocate  open-category positions, and next
for each reserve-eligible category to allocate VR-protected positions.
\smallskip

	\begin{quote}
        \noindent{}{\bf Two-Step Meritorious Horizontal (2SMH) Choice Rule} {\boldmath$C^{2s}_{\circled{M}} = (C^{2s,\nu}_{\circled{M}})_{\nu \in \calv}$}\smallskip

\noindent For every $I \subseteq \cali$,
\begin{eqnarray*}
&& C^{2s,o}_{\circled{M}}(I) = C^{o}_{\circled{M}}(I), \mbox{ and}\\
&& C^{2s,c}_{\circled{M}}(I) = C^{c}_{\circled{M}}\big(I^c\setminus C^{o}_{\circled{M}}(I)\big) \quad \mbox{ for any } c \in \calr.
\end{eqnarray*}
        \end{quote}

We next present our main characterization result, extending our analogous characterization of the 2SMG choice rule
under non-overlapping HR protections to its generalization the 2SMH choice rule under overlapping HR protections.

\begin{theorem}\label{thm:2s-envchar}
A choice rule
\begin{enumerate}
\item maximally accommodates HR protections,
\item satisfies no justified envy,
\item is non-wasteful, and
\item complies with VR protections
\end{enumerate}
if, and only if, it is the 2SMH choice rule $C^{2s}_{\circled{M}}$.
\end{theorem}

In addition to being the only choice rule that satisfies each of the four axioms in Theorem \ref{thm:2s-envchar},
our proposed  2SMH choice rule $C^{2s}_{\circled{M}}$ also satisfies the axiom of incentive compatibility defined in Section \ref{subsec-flaws}.

\begin{proposition} \label{prop:ic}
The 2SMH choice rule $C^{2s}_{\circled{M}}$ satisfies incentive compatibility.
\end{proposition}

\subsection{Related Literature}

Our theoretical analysis of reservation policies differs from its predecessors in two ways:
\begin{enumerate}
\item \textit{concurrent\/} implementation of VR and HR protections, and
\item potentially \textit{overlapping\/} structure of HR protections.
\end{enumerate}
While there is a rich literature on affirmative action policies in India and elsewhere,
our paper is the first one to formally analyze vertical and
horizontal reservation policies when they are implemented concurrently.

There are a number of recent papers on reservation policies, most in the context of school choice.
\cite{abdulson03} study affirmative action policies that limit the number of admitted students of
a given type through hard quotas. \cite{koj12} shows that a policy of limiting the number of majority
students through hard quotas can hurt minority students, the intended beneficiaries. To overcome the detrimental effect of
affirmative action policies based on majority quotas, \cite{hayeyi13} introduce 
policies based on \textit{minority reserves\/}.
In the absence of overlapping reservations,   \cite{echyen12} present an axiomatic characterization
of the minimum guarantee choice rule.
Most recently, \cite{pathak/sonmez/unver/yenmez:20} consider a general model of reservation policies to
balance various ethical principles for pandemic medical resource allocation, although their model is not equipped
to analyze concurrent implementation of vertical and overlapping horizontal reservation policies.

A few papers study the implementation of vertical or (non-overlapping) horizontal reservations individually
in various real-life applications.
These include  \cite{dur_boston} for school choice in Boston, \cite{dur16} for school choice in Chicago, and
\cite{prs1:20} for H-1B visa allocation in the US.
All these models  are applications of the more general model in
\cite{komson16}, where the authors introduce a matching model with \textit{slot-specific priorities\/}.
In contrast, our model is  independent from \cite{komson16}.
Three additional papers on reservation policies include \cite{aygtur16,aygtur17}, where the authors study admissions to engineering colleges in India, and  \cite{aybo16}, where the authors study admissions to Brazilian public universities. While the application in \cite{aygtur16,aygtur17} is closely related to ours, their analysis is independent because not only horizontal reservations are assumed away altogether in these papers, but also
analyses in these papers largely abstract away from the legal requirements in India.\footnote{See
also the discussion of Indian college admissions in \cite[Appendix C.1]{echyen12}.}
In contrast, the presence of horizontal reservations is of key importance for our analysis
that is build on Indian legislation.
The Brazilian affirmative action application studied by  \cite{aybo16} relates to ours in that it also includes multi-dimensional reservation policies,
but unlike our models their application is a special case of \cite{komson16}.
There is, however, one important element in our paper that directly builds on   \cite{aybo16}.
The two desiderata  that play an important role in our proposed reform in India,
\textit{no justified envy\/} and \textit{incentive compatibility\/} are originally introduced by
 \cite{aybo16}.  
Evidence from aggregate data  suggesting 
that the presence of justified envy is widespread in Brazil is also presented in this paper.
As in  \cite{aybo16}, we also present extensive evidence of justified envy in the field, but  in addition we also
document the large scale disruption this anomaly creates in the field. Other less related papers on
reservation policies include
\cite{westkamp10},  \cite{ehayeyi14},
\cite{kamakoji-basic}, and \cite{frapet17}.

In the absence of vertical reservations, 
 analysis of overlapping horizontal  reservations has received some attention in the literature
(\cite{kurata17}), albeit for a different variant of the problem  where individuals have strict preferences
for whether and which protection is invoked in securing a position.
When applied in an environment where individuals are indifferent between all positions,
choice rules recommended in  \cite{kurata17} result in the limitations presented in Section \ref{sec:singlecategory}.
Building on the literature in matroid  theory,  we overcome these difficulties with the meritorious horizontal choice rule. 
More specifically,  Proposition \ref{prop:compenv} and Theorem \ref{thm:envchar} are conceptually related 
to abstract  results in matroid theory. 
Proposition \ref{prop:compenv} can be seen as a generalization of a result
in \cite{gale1968} which shows that the outcome of the Greedy algorithm ``dominates'' any independent set of a matroid. In our
appendix, we refer to this domination relation as ``Gale domination.'' 
The first step of our meritorious horizontal choice rule corresponds to the Greedy algorithm defined on an adequately defined matroid, 
and  Proposition \ref{prop:compenv} shows that this choice rule Gale dominates any choice rule that maximally complies with HR protections.
The proof uses mathematical induction on the number of individuals chosen at the second step of our
choice rule and uses Gale's result for the base case. 
Parts of the proof of Theorem \ref{thm:envchar} uses the properties of the Greedy algorithm. 

More broadly, our paper contributes to the field of market design, where
economists are increasingly taking advantage of advances in
technology to design new or improved allocation mechanisms in
applications as diverse as entry-level labor markets \citep{roth99}, school choice \citep{balson99,abdulson03},
spectrum auctions \citep{milgr00b}, kidney exchange \citep{roth04,roth:sonmez:unver:2005}, internet auctions \citep{edossc07,varian07},
course allocation \citep{sonunv10,budish:2011},  cadet-branch matching \citep{sonmez:switzer:2013,sonmez_rotc2011},
assignment of airline arrival slots \citep{schummer:vohra:2013,schummer:abizada:2017}, and
refugee matching \citep{jones:teyt:2017,del16,ander17}.

\section{Epilogue: Life Imitates Science with the December 2020 Supreme Court judgment \textit{Saurav Yadav v State of Uttar Pradesh (2020)\/} }  \label{sec:Epilogue}
As our paper was under revision for this journal, a December 2020 Supreme Court judgment in
\textit{Saurav Yadav v State of Uttar Pradesh (2020)\/} became headline news in India.\footnote{See, for example,
 \textit{The Indian Express\/} opinion dated 12/26/2020 ``SC verdict exposes fallacy of using general category as reservation for upper castes,'' available in
\url{https://indianexpress.com/article/opinion/columns/casteism-supreme-court-saurav-yadav-reservation-7120348/}, and
\textit{The Wire\/} analysis dated 12/23/2020 ``How the Supreme Court Blocked Attempts to Dilute Merit Under the Open Category,'' available in
\url{https://thewire.in/law/supreme-court-reservation-merit}. Both links last accessed on 06/03/2021.}
Using arguments parallel to our analysis presented in Section \ref{sec:non-overlapping} and
the evidence we documented  from high court cases presented in Section \ref{sec:justifiedenvy} of the Online Appendix,
a three-judge bench of the highest court reached much of the same conclusions we had reached earlier in the March 2019 working version of this paper
in \cite{sonyen19}.
Most notably, similar to our policy recommendations, with this judgment
\begin{enumerate}
\item all allocation rules for public recruitment are federally mandated to satisfy \textit{no justified envy\/}, and thereby
\item the SCI-AKG choice rule, mandated for 25 years, becomes rescinded.
\end{enumerate}
Using several of the same judgments we present in Section \ref{sec:justifiedenvy} of the Online Appendix, the justices also highlighted
the inconsistencies between several high court judgments in relation to desiderata  we formulated as the axiom of \textit{no justified envy\/}.
The justices also declared that while the ``first view'' that enforces \textit{no justified envy\/}
by  the high court judgments  of Rajasthan, Bombay, Gujarat, and Uttarakhand is ``correct and rational,''
 the ``second view'' that  allows for  \textit{justified envy\/} by the high court judgments of Allahabad and
Madhya Pradesh is not.\footnote{It is important to emphasize that, prior to this ruling,
the second view\textemdash now deemed incorrect and irrational\textemdash  was in line with the SCI-AKG choice rule,
whereas the first view\textemdash now deemed correct and rational\textemdash deviated from the previously mandated choice rule.}

While the axiom of \textit{no justified envy\/} becomes federally enforced with \textit{Saurav Yadav v State of Uttar Pradesh (2020)\/},
unlike in \textit{Anil Kumar Gupta (1995)\/} no explicit procedure is mandated with this new Supreme Court ruling.
Two points, however, are important to emphasize in this regard.
The first one is that  prior to \textit{Saurav Yadav (2020)\/},  the 2SMG choice rule became mandated
in the state of Gujarat with the August 2020 high court judgment \textit{Tamannaben Ashokbhai Desai v. Shital Amrutlal
Nishar (2020\/)}.\footnote{The mandated choice rule in Gujarat is described for a
single group of beneficiaries (women) for horizontal reservations
under this High Court ruling. See Section \ref{sec:Gujarat} in the Online Appendix for the description of the procedure in
 \textit{Tamannaben Ashokbhai Desai (2020\/).}}
While the  justices of the Supreme Court have not enforced any specific rule in their December 2020 judgment, they endorsed
the  \textit{2SMG} choice rule given in \textit{Tamannaben Ashokbhai Desai v. Shital Amrutlal Nishar (2020\/)}:
\begin{indquote}
36. Finally, we must say that the steps indicated by the High Court of Gujarat in para 56 of its judgment
in Tamannaben Ashokbhai Desai contemplate the correct and appropriate procedure for considering and
giving effect to both vertical and horizontal reservations. The illustration given by us deals with only one possible dimension.
There could be multiple such possibilities. Even going by the present illustration, the first female candidate allocated
in the vertical column for Scheduled Tribes may have secured higher position than the candidate at Serial No.64.
In that event said candidate must be shifted from the category of Scheduled Tribes to Open / General category causing
a resultant vacancy in the vertical column of Scheduled Tribes. Such vacancy must then enure to the benefit of the
candidate in the Waiting List for Scheduled Tribes - Female.

The steps indicated by Gujarat High Court will take care of every such possibility. It is true that the exercise of laying
down a procedure must necessarily be left to the concerned authorities but we may observe that one set out in said
judgment will certainly satisfy all claims and will not lead to any incongruity as highlighted by us in the preceding paragraphs.
\end{indquote}
Since both the Supreme Court's and the Gujarati High Court's judgments abstract away from any issues in relation to
overlapping horizontal reservations, these rulings are parallel to our recommendations in Section \ref{sec:non-overlapping}.
There is, however, a potentially misleading aspect in the last sentence of the above quote in  \textit{Saurav Yadav (2020)\/},
which brings us to our second point.

Apart from enforcing the axiom of \textit{no justified envy\/} and rescinding the SCI-AKG choice rule, \textit{Saurav Yadav (2020)\/} also
brought clarity to a subtle  aspect of implementation of  vertical reservations in the presence of horizontal reservations.  When the
concept of vertical reservations was
originally introduced in \textit{Indra Sawhney (1992)\/}, positions awarded to individuals selected in the open competition on the
basis of their merit were prohibited from counting against vertically reserved positions. Since then, this aspect of vertical reservations has been used
as its key defining characteristic in India.
However, no judgment of the Supreme Court prior to \textit{Saurav Yadav (2020)\/} explicitly
formulated what it means to get selected in the open competition on the basis of merit in the presence of horizontal reservations.
To a large extent, much of the disarray in India in relation to concurrent implementation of VR and HR policies boils down to this ambiguity.
This important gap is now clarified under \textit{Saurav Yadav (2020)\/}, where an individual who qualifies for
an open-category HR-protected position on the basis of her merit is  explicitly considered as an individual who gets selected  in the open competition on the
basis of merit. This clarification is of key importance, because with the resolution of this ambiguity the 2SMG choice rule
remains the only choice rule by Theorem \ref{thm:2smg} that satisfies all mandates of the Supreme Court for applications in the field with non-overlapping horizontal reservations.
Therefore, while the justices have not explicitly mandated the 2SMG choice rule with \textit{Saurav Yadav (2020)\/}  and they merely endorsed it emphasizing that
``the exercise of laying down a procedure must necessarily be left to the concerned authorities,''
they have indirectly enforced it when individuals have at most one trait.

Finally, while the judgments of the Supreme Court offer some flexibility for the more general case of overlapping
horizontal reservations, we have advocated in Section \ref{sec:overlapping} for a specific choice rule, the \textit{two-step meritorious horizontal\/} choice
rule, for this more general case, and characterized it in Theorem \ref{thm:2s-envchar} with axioms which can be considered natural extensions of the
simpler versions mandated by the Supreme Court.


\bibliographystyle{aer}
\bibliography{matching}

\appendix

\begin{center}
\textbf{\Large Appendix}
\end{center}

\section{Proofs} \label{Appendix:preliminariesandproofs}
In this Appendix, we present the proofs of our results. Some of our results for the more general version of
the model in Section \ref{sec:overlapping}, most notably  Proposition \ref{prop:compenv} and Theorem \ref{thm:envchar},
are conceptually related to abstract results in \textit{matroid theory\/}. Although these results have more
direct proofs that rely on the literature on \textit{maximum matchings in bipartite graphs\/}, we present
proofs that highlight the conceptual connection between our results and the literature on matroid theory.

Before we present the proofs of our results in Section \ref{appendix:proofs}, we present preliminaries 
in matroid theory in Sections \ref{sec:matroid} and \ref{app:greedy}.

\subsection{Preliminary Definitions and Results in Matroid Theory}\label{sec:matroid}
In this section we provide some basic definitions and results in matroid theory.
We follow \cite{oxley}.

A \emph{matroid} is a pair $(E,\calm)$ where $E$ is a finite set and $\calm$ is a
collection of subsets of $E$ that satisfies the following three properties:
\begin{description}
  \item[M1] $\emptyset \in \calm$.
  \item[M2] If $M\in \calm$ and $M'\subseteq M$, then $M'\in \calm$.
  \item[M3] If $M_1, M_2 \in \calm$ and $|M_1|<|M_2|$, then there is $m \in M_2 \setminus M_1$ such that
  $M_1\cup \{m\} \in \calm$.
\end{description}

Set $E$ is called the \emph{ground set} of the matroid. Each set in $\calm$ is called an \emph{independent set}.
An independent set $M$ is \emph{maximal} if there is no proper superset of $M$ that is independent.
A maximal independent set of a matroid is called a \emph{base}. All bases of a matroid
have the same cardinality by M3. The set of bases $\mathcal{B}$ satisfies the following two properties:
\begin{description}
  \item[B1] $\mathcal{B}$ is non-empty.
  \item[B2] If $B_1$ and $B_2$ are in $\mathcal{B}$ and $e_1 \in B_1 \setminus B_2$, then there
  exists an element $e_2$ of $B_2 \setminus B_1$ such that $(B_1\setminus \{e_1\}) \cup \{e_2\} \in \mathcal{B}$.
\end{description}
The stronger version of B2 where the implication is $(B_1\setminus \{e_1\}) \cup \{e_2\} \in \mathcal{B}$ and
$(B_2\setminus \{e_2\}) \cup \{e_1\} \in \mathcal{B}$ also holds \citep{brualdi1969}. An analogous statement
holds when instead of individual elements in $B_1\setminus B_2$ and $B_2\setminus B_1$, we consider sets of elements
\citep{brylawski1973,greene1973,woodall1974}:
\begin{description}
  \item[B2'] If $B_1$ and $B_2$ are in $\mathcal{B}$ and $E_1 \subseteq B_1 \setminus B_2$, then there
  exists $E_2\subseteq B_2 \setminus B_1$ such that $(B_1\setminus E_1) \cup E_2 \in \mathcal{B}$
  and $(B_2\setminus E_2) \cup E_1 \in \mathcal{B}$.
\end{description}

The \emph{restriction} of matroid $(E,\calm)$ to $E' \subseteq E$ is a matroid $(E',\calm')$
where $\calm'=\{X\subseteq E' : X\in \calm\}$. The \emph{rank} of $X\subseteq E$
is defined as the cardinality of a maximal independent set in the restriction of $(E,\calm)$ to $X$.
Since all maximal independent sets have the same cardinality, the rank of a set is well-defined.
The rank of $X\subseteq E$ is is denoted by $r(X)$. The rank function satisfies the following
properties:
\begin{description}
  \item[R1] If $X\subseteq E$, then $0\leq r(X) \leq |X|$.
  \item[R2] If If $X\subseteq Y\subseteq E$, then $r(X) \leq r(Y)$.
  \item[R3] If $X,Y \subseteq E$, then
            \[r(X\cup Y)+r(X\cap Y)\leq r(X)+r(Y).\]
  \end{description}

\subsection{Greedy Choice Rule and Its Properties} \label{app:greedy}

For a given weight function $w:E\rightarrow \mathbb{R}_+$ that takes distinct values, the
\emph{greedy algorithm} chooses the element with the highest weight subject to the constraint
that the chosen set of elements is independent. 

	\begin{quote}
        \noindent{}{\bf Greedy Algorithm} \smallskip

		\noindent{}{\bf Step 1}:
			Set $X_0=\emptyset$ and $i=0$.
	
		\noindent{}{\bf Step 2}:
		If there exists $e \in E\setminus X_i$ such that $X_i \cup \{e\} \in \calm$, then
        choose such an element $e_{i+1}$ of maximum weight, let $X_{i+1}=X_i \cup \{e_{i+1}\}$,
        and go to Step 3; otherwise let $B=X_i$ and go to Step 4.

        \noindent{}{\bf Step 3}:
           Add 1 to $i$ and go to Step 2.

        \noindent{}{\bf Step 4}:
           Stop.
	\end{quote}

The textbook definition of the Greedy algorithm takes $w$ to be any weight function that can take
same values for different elements of $E$. In this case, the Greedy algorithm can
select different sets depending on how elements are chosen when they have the same weight.
To avoid this issue, we assume that distinct elements of $E$ have different weights.

The greedy algorithm is defined on matroid $(E,\calm)$. However, it can be applied
to any restriction of this matroid. Therefore, the greedy algorithm can be viewed as
a single-category choice rule on $2^E$ \citep{fleiner2001}. For the rest of the paper,
we view it as a single-category choice rule and refer to it as the \emph{greedy choice rule}.

The greedy algorithm chooses an independent set that has the maximum weight, where
the weight of a set is the sum of weights of individual elements. Before we introduce
a stronger property of the greedy algorithm, we need the following definition.

Let elements of the sets $X,Y \subseteq E$ be enumerated such that,
\begin{align*}
\mbox{for every} \; i,j\in \{1,\ldots,|X|\}, \qquad  i \leq j \; &\implies \; w(x_i) \geq w(x_j), \; \mbox{ and}\\
\mbox{for every} \; i,j\in \{1,\ldots,|Y|\}, \qquad  i \leq j \; &\implies \; w(y_i) \geq w(y_j).
\end{align*}
Then,  the set $X=\{x_1,\ldots,x_{|X|}\} \subseteq E$ \emph{Gale dominates} the set $Y=\{y_1,\ldots,y_{|Y|}\} \subseteq E$ if $|X|\geq |Y|$ and,
for every $i\in \{1,\ldots,|Y|\}$,
\[
w(x_i) \geq w(y_i).
\]
We use the notation $X \succeq^G Y$ to denote set $X$ Gale dominates set $Y$.\smallskip

The following property of the greedy choice rule is the driving force for a similar property
of the meritorious horizontal choice rule that is presented in
Proposition \ref{prop:compenv}.

\begin{lemma}\citep{gale1968} \label{lem:Gale}
For every $E' \subseteq E$, the outcome of the
greedy choice rule for $E'$ Gale dominates any independent subset of $E'$.
\end{lemma}

The following property of choice rules plays an important role in market design.
\begin{definition}\citep{kelso82}
A choice rule $C:2^{E} \rightarrow 2^{E}$ satisfies the \textbf{substitutes} condition, if, for
every $E' \subseteq E$,
\[e\in C(E') \text{ and } e' \in E'\setminus \{e\} \; \implies \;  e\in C(E' \setminus \{e'\}).\]
\end{definition}

We use the following result in some of our proofs.

\begin{lemma}\citep{fleiner2001}\label{lem:subst}
The greedy choice rule satisfies the substitutes condition.
\end{lemma}

Fix a matroid $(E,\calm)$ with rank function $r$. We next formulate some properties of choice rules.
The first two properties extend the notion of independence for sets and the maximality for independent sets to choice rules.

\begin{definition}
A choice rule $C:2^{E} \rightarrow 2^{E}$ is \textbf{independent} if, for every $E'\subseteq E$,
$C(E')$ is an independent set.
\end{definition}

\begin{definition}
A choice rule $C:2^{E} \rightarrow 2^{E}$ is \textbf{rank maximal} if, for every $E'\subseteq E$,
\begin{center}
$r(C(E'))=r(E')$.
\end{center}
\end{definition}

The next property is a reformulation of our axiom \textit{no justified envy\/} in the abstract context of matroids.

\begin{definition}\label{def:matjustified}
A choice rule $C:2^{E} \rightarrow 2^{E}$ satisfies \textbf{no justified envy} if, for every $E'\subseteq E$,
$e\in C(E')$, and $e'\in E'\setminus C(E')$
\[ w(e') \, > \, w(e)    \; \implies \;  r\big((C(E')\setminus \{e\}) \cup \{e'\}\big) < r(C(E')).
\]
\end{definition}

The following result follows from the well-known properties of the greedy algorithm. We will rely on 
it extensively in the proof of Theorem  \ref {thm:envchar} to highlight the conceptual similarities 
between our characterization of the meritorious horizontal choice rule and some of the key properties 
of the greedy algorithm.

\begin{lemma}\label{lem:greedy}
A choice rule $C:2^{E} \rightarrow 2^{E}$ is independent, rank maximal, and satisfies no justified envy if, and only if, it is the greedy choice rule.
\end{lemma}

\begin{proof}
Let $C$ be the greedy choice rule. Then by construction it is independent. Rank maximality follows easily because if
$C(E')$ is not rank maximal for some $E'\subseteq E$, then there exists $e\in E$ such that $C(E')\cup \{e\}$
is independent. Thus, the greedy choice rule cannot produce $C(E')$.

Suppose, for contradiction, that $C$ fails to satisfy no justified envy. Then there exists $E'\subseteq E$,
$e\in C(E')$, and $e'\in E'\setminus C(E')$ with $w(e')>w(e)$ such that
\[r\big((C(E')\setminus \{e\}) \cup \{e'\}\big) \geq r(C(E')).\]
Since $C$ is rank maximal, $r(C(E'))=r(E')$. By R2, $r((C(E')\setminus \{e\}) \cup \{e'\} ) \leq r(E')$.
Therefore, $r(C(E')\setminus \{e\} \cup \{e'\})=r(E')$. Furthermore, since $|C(E')|=r(C(E'))$ because $C(E')$
is independent, we get $r(C(E')\setminus \{e\} \cup \{e'\})=|C(E')\setminus \{e\} \cup \{e'\}|$, so
$(C(E')\setminus \{e\}) \cup \{e'\}$ is also an independent set.
By Lemma \ref{lem:Gale}, $C(E')$ Gale dominates $(C(E')\setminus \{e\}) \cup \{e'\}$, so we get
$w(e)>w(e')$, which is a contradiction.

We next show that any choice rule satisfying the properties has to be the greedy choice rule. Let $D$
be a choice rule that satisfies the three axioms.
Suppose, for contradiction, that $D(E')\neq C(E')$ for some $E'\subseteq E$. Since both
$D$ and $C$ are independent and rank maximal, $D(E')$ and
$C(E')$ are bases in the matroid restriction of $(E,\calm)$ to $E'$ and so $|D(E')|=|C(E')|$.
Therefore, there exists $e_1$ in $D(E')\setminus C(E')$. Then, by B2',
there exists $e_2 \in C(E')\setminus D(E')$ such that
$(D(E')\setminus \{e_1\})\cup \{e_2\}$ and $(C(E')\setminus \{e_2\}) \cup \{e_1\}$ are also bases.
By Lemma \ref{lem:Gale}, $C(E')$ Gale dominates $(C(E')\setminus \{e_2\}) \cup \{e_1\}$, which implies
that $w(e_2)>w(e_1)$. Since $D$ satisfies no justified envy, $e_1\in D(E')$, $e_2\in E' \setminus D(E')$,
and $w(e_2)>w(e_1)$, we get
\[r\big((D(E')\setminus \{e_1\}) \cup \{e_2\}\big) < r(D(E')).\]
This is a contradiction because $(D(E')\setminus \{e_1\}) \cup \{e_2\}$ and $D(E')$ are
both bases and, therefore, they have the same rank since they have the same cardinality.
\end{proof}

\subsection{Proofs of Main Results}\label{appendix:proofs}
Using the HR graph for a category $v\in \calv$, we can study the \emph{transversal matroid} with the ground set
$\cali^v$ \citep{edmondsfulkerson65}. In this matroid, a set of individuals is
independent if they can be matched with distinct positions and therefore the
rank of a set of individuals is equal to the maximum number of distinct positions
they can be matched with, which is the $n^v$ function that we have defined for a category $v\in \calv$.
Furthermore, the weight of an individual can be defined as their merit score. With this setup, Step 1 of the meritorious
horizontal choice rule is the same as the greedy choice rule for the transversal matroid.
We use this important observation in proofs of Proposition  \ref{prop:compenv} and Theorem \ref{thm:envchar} presented below.

\subsection*{Proof of Proposition \ref{prop:SCI-vs-2smg}}
Let $I\subseteq \cali$ be a set of individuals and $I^m\subseteq I$ be the set of reserve-eligible
individuals considered at Step 1 of $\widehat{C}^{SCI}$ when $I$ is the set of applicants.

Let $i\in \widehat{C}^{2s}_{mg}(I)\cap I^g$. Then $i\in C^o_{mg}(I) \cap I^g$ because
$\widehat{C}^{2s}_{mg}(I) \cap I^g=C^o_{mg}(I) \cap I^g$. Since
$C^o_{mg}$ satisfies the substitutes condition \citep{echyen12}, $i\in C^o_{mg}(I^m \cup I^g)$ because $i\in I^g$ and
$i\in C^o_{mg}(I)$. Therefore, $i \in C^o_{mg}(I^m \cup I^g) \cap I^g$, which implies
$i \in \widehat{C}^{SCI}(I) \cap I^g$ because $\widehat{C}^{SCI}(I) \cap I^g =C^o_{mg}(I^m \cup I^g) \cap I^g$.
Therefore, we conclude that $\widehat{C}^{2s}_{mg}(I)\cap I^g \subseteq \widehat{C}^{SCI}(I) \cap I^g$.

The assumption that $|I^c|\geq q^o+q^c$ for each reserve-eligible category $c\in \calr$ implies that
all category-$c$ positions are filled under both $C^{2s}_{mg}$ and $C^{SCI}$. In addition, the
first part of the proposition implies that there are weakly more individuals with reserved categories assigned
to open-category positions under $C^{2s}_{mg}$ than under $C^{SCI}$. Therefore,
\[  \sum_{c\in\calr} \big|\widehat{C}^{2s}_{mg}(I) \cap I^c\big| \geq \sum_{c\in\calr} \big|\widehat{C}^{SCI}(I) \cap I^c\big|.
\]

\qed

\medskip

\subsection*{Proof of Theorem \ref{thm:2smg}} Since the 2SMH choice rule reduces to  the 2SMG choice rule
in the absence of overlapping horizontal reservations, the result is a direct corollary of  Theorem \ref{thm:2s-envchar}. \qed

\medskip

\subsection*{Proof of Proposition \ref{prop:compenv}}
Let $I\subseteq \cali^v$ be a set of individuals. To show part (1), note that $|\cenv ^v (I)|=\min\{q^v,|I|\}$
since $\cenv ^v$ is non-wasteful.
Furthermore, for single-category choice rule $C^v$, $C^v(I) \subseteq I$ and $|C^v(I)|\leq q^v$ imply
\[|C^v(I)|\leq \min\{q^v,|I|\} = |\cenv^v(I)|. \]

We show the second part using mathematical induction on the number of individuals chosen at the second step
of $\cenv^v$. For the inductive step, we decrease $q^v$ by one so one less individual
is chosen at the second step of $\cenv^v$ and we also consider a subset of $I$.

For the base case, when no individuals are chosen at the second step of $\cenv^v$, Lemma \ref{lem:Gale} states
$\cenv^v(I) \succeq^G C^v(I)$ since the first step of $\cenv^v$ is the greedy choice rule for the transversal matroid and
$C^v(I)$ is another base since $C^v$ maximally accommodates HR protections and $|\cenv^v(I)|=|C^v(I)|$ in this case.

Now assume that the claim holds when the number of individuals chosen at the second step of $\cenv^v$ is
less than $k>0$. Let $C^v_R$ be the choice rule corresponding to the second step of $\cenv^v$. Consider a set of
individuals $I$ such that $|C^v_R(I)|=k$.
Let $J=\cenv^v(I)$, $K=C^v_R(I)$, $J'=\cenv^v(C^v(I))$, and $K'=C^v(I) \setminus J'$.
If $|K'|=0$, then the proof is complete as in the base case using Lemma \ref{lem:Gale}. For the rest of
the proof suppose that $|K'|>0$.

\begin{lemma} \label{Lemma4}
There exists $j\in K$ and $j'\in K'$ such that $\sigma(j)\geq \sigma(j')$.
\end{lemma}

\begin{proof}\renewcommand{\qedsymbol}{$\blacksquare$}
Suppose, for contradiction, that for every $j\in K$ and $j'\in K'$
we have $\sigma(j)<\sigma(j')$. Since $j\in K=C^v_R(I)$, $j'\notin K$, and $\sigma(j')>\sigma(j)$,
we must have $j'\in J$. Therefore, every individual in $K'$ is also in $J$, which means $K'\subseteq J$.
Since $K'\cap J'=\emptyset$ we have $K'\subseteq J\setminus J'$. Therefore, by B2' there exists
$K''\subseteq J'\setminus J$ such that $(J\setminus K')\cup K''$ and $(J'\setminus K'')\cup K'$
are also bases. By Lemma \ref{lem:Gale} we get
\[J \succeq^G (J \setminus K') \cup K''\]
and
\[J' \succeq^G (J' \setminus K'') \cup K'.\]
These are equivalent to $K'\succeq^G K''$ and $K'' \succeq^G K'$, respectively.
Therefore, we must have $K'=K''$, which is a contradiction because $K'\neq \emptyset$
and $K'\cap K''=\emptyset$.
\end{proof}

Now we use Lemma \ref{Lemma4} to finish the proof. Consider $j\in K$ and $j'\in K'$ such that $\sigma(j)\geq \sigma(j')$.
If $\sigma(j)=\sigma(j')$, then $j=j'$ since different individuals have distinct merit scores. In this case,
consider the set of individuals $I\setminus \{j\}$ and reduce the capacity $q^v$ to $\min\{q^v,|\cenv^v(I)|\}-1$.
For $\cenv^v$, the same set of individuals is chosen at Step 1 by the greedy choice rule and the same set
of individuals except $j$ is chosen at the second step. Consider a choice rule $D^v$ such that
$D^v(I\setminus \{j\})=C^v(I)\setminus \{j\}$ and, for other $I'\neq I$, let $D^v(I')$ be such
that $n^v(D^v(I'))=n^v(I')$. Then $D^v$ maximally accommodates HR protections. By the mathematical induction
hypothesis, we get that $\cenv^v(I)\setminus \{j\} \succeq^G D^v(I\setminus \{j\})=C^v(I)\setminus \{j\}$. Therefore,
$\cenv^v(I)\succeq^G C^v(I)$ because $j\in \cenv^v(I)$ and $j\in C^v(I)$.\smallskip

Next consider the case when $\sigma(j)>\sigma(j')$. We need the following result.

\begin{lemma}
Individual $j'$ is not a member of $J$.
\end{lemma}

\begin{proof}
Suppose, for contradiction, that $j'\in J$. Since $j'\in K'$, we have $j'\notin J'$. Therefore,
$j'\in J\setminus J'$. By B2', there exists $j''\in J'\setminus J$ such that both
$(J \setminus \{j'\}) \cup \{j''\}$ and $(J' \setminus \{j''\}) \cup \{j'\}$ are bases.
By Lemma \ref{lem:Gale}, $J \succeq^G (J \setminus \{j'\}) \cup \{j''\}$ and
$J' \succeq^G (J' \setminus \{j''\}) \cup \{j'\}$, which are equivalent to
$\{j'\} \succeq^G  \{j''\}$ and $\{j''\} \succeq^G  \{j'\}$, respectively.
The last two inequalities can only hold when $j'=j''$, which is a contradiction
because $j'\in K'$, $j''\in J'$, and $J'\cap J''=\emptyset$.
\end{proof}

We apply the inductive hypothesis to the market with the set of individuals $I \setminus \{j,j'\}$ and the capacity
$\min\{q^v,|\cenv^v(I)|\}-1$ as in the previous case (when $j=j'$). In this reduced market, at the first step of
$\cenv^v$, the greedy choice rule selects the same set of individuals as in the original market and
the responsive choice rule selects the same set of individuals except $j$. Construct choice rule
$D^v$ such that $D^v(I\setminus \{j,j'\})=C^v(I)\setminus \{j'\}$. For any other $I'\neq I$ let
$D^v(I')$ be such that $n^v(D^v(I'))=n^v(I')$. Since $D^v(I\setminus \{j,j'\}) \supseteq J'$,
$n^v(D^v(I\setminus \{j,j'\}))\geq n^v(J')=n^v(I)$, where the inequality follows from monotonicity
of $n^v$ and the equality follows since $C^v$ maximally satisfies HR protections. Therefore,
$n^v(D^v(I\setminus \{j,j'\}))=n^v(I\setminus \{j,j'\})$, and hence $D^v$ maximally accommodates HR protections.
By the mathematical induction hypothesis, $\cenv^v(I)\setminus \{j\} \succeq^G C^v(I)\setminus \{j'\}$. Furthermore,
since $\sigma(j)>\sigma(j')$, we conclude that $\cenv^v(I) \succeq^G C^v(I)$.

\qed

\medskip

\subsection*{Proof of Theorem \ref{thm:envchar}}
We first show $\cenv^v$ satisfies the stated properties in several lemmas and then show that the
unique category-$v$ choice rule satisfying these properties is $\cenv^v$. Let
$C^v_G$ be the greedy choice rule that corresponds to the first step of $\cenv^v$. Let
$C^v_R$ be the choice rule that corresponds to the second step of $\cenv^v$: For every
$I\subseteq \cali^v$, $C^v_R(I)$ consists of $\min\{q^v-|C^v_G(I)|,|I|-|C^v_G(I)|$ individuals
in $I\setminus C^v_G(I)$ with the highest merit scores. Therefore, we have
\[\cenv^v(I)=C^v_G(I) \cup C^v_R(I).\]

\begin{lemma}\label{lem:compliance}
$\cenv^v$ maximally accommodates HR protections.
\end{lemma}

\begin{proof}\renewcommand{\qedsymbol}{$\blacksquare$}
For every $I\subseteq \cali^v$, by Lemma \ref{lem:greedy}, $n^v(C_G^v(I))=n^v(I)$.
Furthermore, by monotonicity of $n^v$, $n^v(\cenv^v(I))\geq n^v(C_G^v(I))=n^v(I)$,
which implies $n^v(\cenv^v(I))=n^v(I)$. We conclude that $\cenv^v$ maximally accommodates HR protections.
\end{proof}

\begin{lemma}
$\cenv^v$ satisfies no justified envy.
\end{lemma}

\begin{proof}\renewcommand{\qedsymbol}{$\blacksquare$}
Suppose, for contradiction, that $\cenv^v$ fails no justified envy. Therefore, there exist
a set of individuals $I \subseteq \cali^v$, individuals $i\in \cenv^v(I)$, $j \in I \setminus \cenv^v(I)$ with
$\sigma(j) > \sigma(i)$ and $n^v\left((\cenv^v(I)\setminus \{i\}) \cup \{j\}\right) \geq n^v(\cenv^v(I))$.
Since $\cenv^v(I)$ maximally accommodates HR protections, the last inequality implies
$n^v((\cenv^v(I)\setminus \{i\}) \cup \{j\})=n^v(I)$.

Since every individual in $C^v_R(I)$ has a higher merit score than $j$, we must have $i\in C^v_G(I)$. Furthermore,
every individual in $C^v_R(I)$ has a higher merit score than $i$ as well.

Let $I_1=C^v_G((\cenv^v(I)\setminus \{i\}) \cup \{j\})$. Since $C^v_G$ is rank maximal (Lemma \ref{lem:greedy}),
$n^v(I_1)=n^v(I)$. Furthermore, since $C^v_G$ is independent (Lemma \ref{lem:greedy}), $I_1$ is indepedent. Therefore,
we get $|I_1|=n^v(I)$. In addition, since $C^v_G$ satisfies the substitutes condition (Lemma \ref{lem:subst}),
$I_1 \supseteq C^v_G(I)\setminus \{i\}$. Therefore, $I_1=(C^v_G(I)\setminus \{i\}) \cup \{k\}$ where
$k \in C^v_R(I)\cup \{j\}$. This gives us a contradiction since $I_1$ is an independent set, so
by Lemma \ref{lem:Gale}, $C^v_G(I) \succeq^{G} I_1$ which is equivalent to $\sigma(i)>\sigma(k)$
but every individual in $C^v_R(I)\cup \{j\}$ has a higher merit score than $i$.
\end{proof}

\begin{lemma}
$\cenv^v$ is non-wasteful.
\end{lemma}

\begin{proof}\renewcommand{\qedsymbol}{$\blacksquare$}
$\cenv^v$ is non-wasteful because at the second step all the unfilled positions
are filled with the unmatched individuals until all positions are filled or all individuals are assigned
to positions.
\end{proof}

\begin{lemma}
If a category-$v$ choice rule maximally accommodates HR protections,
satisfies no justified envy, and is non-wasteful, then it has to be $\cenv^v$.
\end{lemma}

\begin{proof}\renewcommand{\qedsymbol}{$\blacksquare$}
Let $C^v$ be a category-$v$ choice rule that maximally accommodates HR protections,
satisfies no justified envy, and is non-wasteful. Construct the following choice rule
$D$, where for any $I\subseteq \cali^v$,
\[D^v(I)=C^v_G(C^v(I)).\]
We show that $D^v$ is the greedy choice rule by showing that $D^v$ is independent,
rank maximal, and satisfies no justified envy.

First, $D^v$ is independent because $C^v_G$ is independent (Lemma \ref{lem:greedy}).

Second, since $C^v$ maximally accommodates HR protections, $n^v(C^v(I))=n^v(I)$. In addition,
since the greedy choice rule is rank maximal $n^v(D^v(I))=n^v(C^v(I))$. Therefore,
$n^v(D^v(I))=n^v(I)$, which means that $D^v$ is rank maximal.

Finally, suppose for contradiction that $D^v$ fails to satisfy no justified envy. Then there
exist $I\subseteq \cali$, $i\in D^v(I)$, and $j\in I \setminus D^v(I)$ with $\sigma(j)>\sigma(i)$
such that
\[n^v\big((D^v(I)\setminus \{i\}) \cup \{j\}\big) \geq n^v(D^v(I)).\]
Since $C^v_G$ satisfies no justified envy (Lemma \ref{lem:greedy}), $j$ has to be in $I\setminus C^v(I)$.
Furthermore $n^v(D^v(I))=n^v(I)$ by rank maximality of $D^v(I)$, so we get
$n^v((D^v(I)\setminus \{i\}) \cup \{j\})=n^v(I)$. As a result,
$n^v((C^v(I)\setminus \{i\}) \cup \{j\})=n^v(I)$ as well.
This gives a contradiction to the assumption that $C^v$ satisfies no justified envy
because $i\in C^v(I)$, $j\in I\setminus C^v(I)$, $\sigma(j)>\sigma(i)$, and
$n^v((C^v(I)\setminus \{i\}) \cup \{j\})=n^v(I)=n^v(C^v(I))$.

Since $D^v(I)$ is independent, rank maximal, and satisfies no justified envy, we conclude
by Lemma \ref{lem:greedy} that $D^v=C_G^v$. Now consider $C^v(I)\setminus C_G^v(I)$.
Since $C^v$ is non-wasteful, $|C^v(I)\setminus C_G^v(I)|=\min\{q^v-C_G^v(I),|I|-C_G^v(I)\}$.
Furthermore, by no justified envy, there cannot be an individual in $I\setminus C^v(I)$ who has a higher
merit score than any individual in $C^v(I)\setminus C_G^v(I)$. Therefore, we get
\[C^v(I)\setminus C^v_G(I) = C^v_R(I).\]
Since $C^v(I)\supseteq C^v_G(I)$, we conclude that $C^v(I)=C^v_G(I)\cup C^v_R(I)$ is the meritorious horizontal choice rule.

\end{proof}
This finishes the proof of Theorem \ref{thm:envchar}.
\qed
\medskip

\subsection*{Proof of Theorem \ref{thm:2s-envchar}}
Let $C=(C^v)_{v\in \calv}$ be a choice rule that complies
with VR protections, maximally accommodates HR protections,
satisfies no justified envy, and is non-wasteful.
We show this result using the following lemmas.

\begin{lemma}
$C^o=\cenv^{2s,o}$.
\end{lemma}

\begin{proof}\renewcommand{\qedsymbol}{$\blacksquare$}
We prove that $C^o$ maximally accommodates category-$o$ HR protections,
satisfies no justified envy, and is non-wasteful.

First, we show that $C^o$ maximally accommodates category-$o$ HR protections.
Suppose, for contradiction, that $n^o(C^o(I))<n^o(I)$ for some $I\subseteq \cali$.
Then there exists $i\in I\setminus C^o(I)$ such that
$n^o(C^o(I)\cup \{i\})=n^o(C^o(I))+1$. If $i\in I\setminus \widehat C(I)$, then
we get a contradiction with the assumption that $C$ maximally accommodates HR protections.
Otherwise, if $i\in C^c(I)$ where $c\in \calr$, then we get
a contradiction with the assumption that $C$ complies with VR protections.
Therefore, $C^o$ maximally accommodates category-$o$ HR protections.

Next, we show that $C^o$ satisfies no justified envy.
Let $i\in C^o(I)$ and $j \in I\setminus C^o(I)$
such that $\sigma(j) > \sigma(i)$. If $j \in I \setminus \widehat C(I)$, then
\[n^o\left((C^o(I)\setminus \{i\})\cup \{j\}\right)<n^o\left(C^o(I)\right)\]
because $C$ satisfies no justified envy. However, if $i\in C^c(I)$ for category $c\in \calr$,
then
\[n^o\left((C^o(I)\setminus \{i\})\cup \{j\}\right)<n^o\left(C^o(I)\right)\]
because $C$ complies with VR protections. Therefore, $C^o$ satisfies no
justified envy.

Now, we show that $C^o$ is non-wasteful, which means that
$|C^o(I)|=\min\{|I|,q^o\}$ for every $I\subseteq \cali$.
If there exists an individual $i\in I$ such that $i\notin \widehat C(I)$,
then $|C^o(I)|=q^o$ because $C$ is non-wasteful. If there exists an
individual $i\in I$ such that $i\in C^c(I)$ where $c=\rho(i)\in \calr$, then $|C^o(I)|=q^o$ because
$C$ complies with VR protections. If these two conditions do not hold, then all the
individuals are allocated open-category positions, i.e., $I=C^o(I)$. Therefore,
under all possibilities, we get $|C^o(I)|=\min\{|I|,q^o\}$, which means
that $C^o$ is non-wasteful.

Since $C^o$ maximally accommodates category-$o$ HR protections, satisfies no justified envy,
and is non-wasteful, we get $C^o=\cenv^o$ (Theorem \ref{thm:envchar}), and hence
$C^o=\cenv^{2s,o}$.
\end{proof}

Let $c \in \calr$, $I\subseteq \cali$, and $\bar I^c = \{i\in I\setminus \cenv^o(I)|\rho(i)=c\}$.

\begin{lemma}
$C^c(I)$ maximally accommodates category-$c$ HR protections for $\bar I^c$.
\end{lemma}

\begin{proof}\renewcommand{\qedsymbol}{$\blacksquare$}
Suppose, for contradiction, that $n^c(C^c(I))<n^c(\bar I^c)$.
This is equivalent to
\[n^c(C^c(I))<n^c(\bar I^c)=n^c\Big(C^c(I) \cup \{i\in I\setminus \widehat C(I)|\rho(i)=c\}\Big),\]
which implies that there exists $i \in I\setminus \widehat C(I)$ who is eligible for category $c$
such that
\[n^c(C^c(I\cup \{i\}))=n^c(C^c(I))+1.\]
This equation contradicts the assumption that $C$ maximally accommodates HR protections.
 Therefore, $C^c(I)$ maximally accommodates category-$c$ HR protections for $\bar I^c$.
\end{proof}

\begin{lemma}
$C^c(I)$ satisfies no justified envy for $\bar I^c$.
\end{lemma}

\begin{proof}\renewcommand{\qedsymbol}{$\blacksquare$}
Let $i\in C^c(I)$ and $j \in \bar I^c\setminus C^c(\bar I^c)$ be such that $\sigma(j)>\sigma(i)$.
Note that $i\in \bar I^c$.
Since $C$ satisfies no justified envy, we have
\[n^c\left(C^c(I)\right) > n^c\left((C^c(I)\setminus \{j\})\cup \{i\}\right).\]
Hence, $C^c$ satisfies no justified envy for $\bar I^c$.
\end{proof}

\begin{lemma}
$|C^c(I)|=\min \{|\bar I^c|,q^c\}$.
\end{lemma}

\begin{proof}\renewcommand{\qedsymbol}{$\blacksquare$}
We consider two cases. First, if $C^c(I)=\bar I^c$, then $|C^c(I)|=\min \{|\bar I^c|,q^c\}$
because $|C^c(I)|\leq q^c$. Otherwise, if $C^c(I) \neq \bar I^c$, then there exists
$i \in \bar I^c\setminus C^c(I)$. Therefore, $i\in I \setminus \widehat C(I)$.
Since $C$ is non-wasteful, we get $|C^c(I)|=q^c$
Since $i\in \bar I^c \setminus C^c(I)$ and $|C^c(I)|=q^c$, $|\bar I^c|>q^c$.
Therefore,   $|C^c(I)|=q^c=\min\{|\bar I^c|,q^c\}$.
\end{proof}

Therefore, $C^c(I)$ maximally accommodates category-$c$ HR protections for $\bar I^c$,
$C^c(I)$ satisfies no justified envy for $\bar I^c$, and $C^c(I)$ is non-wasteful for $\bar I^c$. By
Theorem \ref{thm:envchar}, $C^c(I)=\cenv^c(\bar I^c)$ and, thus,
\[C^c(I)=\cenv^{c}(\bar I^c)  =  \cenv^{c}(\{i\in I\setminus \cenv^o(I)|\rho(i)=c\})
= \cenv^{2s,c}(I).\]

\qed

\medskip

\subsection*{Proof of Proposition \ref{prop:ic}}
Suppose that $i$ is chosen by $\hc2s$ when she withholds some of her reserve-eligible privileges. If
$i$ is chosen by $\cenv^o$ for an open-category position, then $i$ will still be chosen by declaring
all her reserve-eligible privileges because $\cenv^o$ does not use the category information of individuals
and an individual can never benefit from not declaring some of her traits under $\cenv^o$ because she
will have more edges in the category-$o$ HR graph. Otherwise, if $i$ is chosen by $\cenv^c$
where $\rho(i)=c\in \calr$ then she must have declared her reserve-eligible category $c$. In addition,
by declaring all her traits she will still be chosen by $\cenv^c$ if she is not chosen before for
the open-category positions because she will have more edges in the HR graph for category-$c$ positions.

\qed

\medskip

\newpage

\begin{center}
\textbf{\Large Online Appendix}
\end{center}

\section{Institutional Background on Vertical and Horizontal Reservations} \label{sub:Institutional}

In this section of the Online Appendix, we present:
\begin{enumerate}
\item  the description of the concepts of vertical reservation and horizontal
reservation as they are quoted in the Supreme Court judgments  \textit{Indra Sawhney (1992)\/} and
\textit{Rajesh Kumar Daria (2007)\/} in Sections \ref{sub:VH} and \ref{sec:RKD2007},
\item the main quotes from the Supreme Court judgments  \textit{Anil Kumar Gupta (1995)\/} and
\textit{Rajesh Kumar Daria (2007)\/} that allow us to formulate the SCI-AKG choice rule in Section \ref{sec:AKG95},
\item the revised mandates of the Supreme Court judgment \textit{Saurav Yadav (2020)}, which imply the 2SMG
choice rule is the only mechanism that remains lawful for the case of non-overlapping horizontal reservations in Section \ref{app-yadav}, and
\item the description of the 2SMG choice rule that is mandated in the State of Gujarat
as it is quoted in the August 2020
High Court of Gujarat judgment \textit{Tamannaben Ashokbhai Desai  (2020)\/} in Section \ref{sec:Gujarat}.
\end{enumerate}

\subsection{Indra Sawhney (1992): Introduction of Vertical and Horizontal Reservations} \label{sub:VH}

The terms \emph{vertical reservation} and \emph{horizontal reservation} are coined by the Constitution bench of the Supreme Court of India, in
 the historical judgment  \textit{Indra Sawhney (1992)\/}, where
 \begin{itemize}
\item  the former was formulated as a  policy tool to accommodate the higher-level protective provisions sanctioned by Article 16(4) of the Constitution of India, and
\item the latter was formulated as a  policy tool to accommodate the lower-level protective provisions sanctioned by Article 16(1)
of the Constitution of India.
\end{itemize}
The description of these two affirmative action policies and how they are intended  to
interact with each other is given in the judgment with following quote:\smallskip
\begin{indquote}
A little clarification is in order at this juncture: all reservations are not of the same nature.
There are two types of reservations, which may,
for the sake of convenience, be referred to as `vertical reservations' and `horizontal reservations'. The reservation in favour of scheduled castes,
scheduled tribes and other backward classes [under Article 16(4)] may be called vertical reservations whereas reservations in favour of physically
handicapped [under clause (1) of Article 16] can be referred to as horizontal reservations. Horizontal reservations cut across the
vertical reservations -- what is called interlocking reservations. To be more precise, suppose 3\% of the vacancies are reserved in favour of physically handicapped persons;
this would be a reservation relatable to clause (1) of Article 16. The persons selected against his quota will be placed in the appropriate category;
if he belongs to SC category he will be placed in that quota by making necessary adjustments; similarly, if he belongs to open competition
(OC) category, he will be placed in that category by making necessary adjustments.
\smallskip
\end{indquote}
It is further emphasized in the judgment that
vertical reservations in favor of backward classes SC, ST, and OBC  (which the judges refer to
as \emph{reservations proper}) are ``set aside'' for these classes.
\begin{indquote}
In this connection it is well to remember that the reservations under Article 16(4) do not operate like a communal reservation.
It may well happen that some members belonging to, say Scheduled Castes get selected in the open competition field on the basis of their own merit;
they will not be counted against the quota reserved for Scheduled Castes; they will be treated as open competition candidates.\smallskip
\end{indquote}

\subsection{Rajesh Kumar Daria (2007): The Distinction Between Vertical Reservation and Horizontal Reservation} \label{sec:RKD2007}
The distinction between vertical reservations and horizontal reservations, i.e. the ``over-and-above'' aspect of the former and
the ``minimum guarantee'' aspect of the latter,
is further elaborated in the Supreme Court judgment
\textit{Rajesh Kumar Daria (2007)\/}.

\begin{indquote}
The second relates to the difference between the nature of vertical reservation and horizontal reservation. Social reservations in favour of SC, ST and OBC under Article 16(4) are `vertical reservations'. Special reservations in favour of physically handicapped, women etc., under Articles 16(1) or 15(3) are `horizontal reservations'. Where a vertical reservation is made in favour of a backward class under Article 16(4), the candidates belonging to such backward class, may compete for non-reserved posts and if they are appointed to the non-reserved posts on their own merit, their numbers will not be counted against the quota reserved for the respective backward class. Therefore, if the number of SC candidates, who by their own merit, get selected to open competition vacancies, equals or even exceeds the percentage of posts reserved for SC candidates, it cannot be said the reservation quota for SCs has been filled. The entire reservation quota will be intact and available in addition to those selected under Open Competition category. [Vide - Indira Sawhney (Supra), R. K. Sabharwal vs. State of Punjab (1995 (2) SCC 745), Union of India vs. Virpal Singh Chauvan (1995 (6) SCC 684 and Ritesh R. Sah vs. Dr. Y. L. Yamul (1996 (3) SCC 253)]. But the aforesaid principle applicable to vertical (social) reservations will not apply to horizontal (special) reservations. Where a special reservation for women is provided within the social reservation for Scheduled Castes, the proper procedure is first to fill up the quota for scheduled castes in order of merit and then find out the number of candidates among them who belong to the special reservation group of `Scheduled Castes-Women'. If the number of women in such list is equal to or more than the number of special reservation quota, then there is no need for further selection towards the special reservation quota. Only if there is any shortfall, the requisite number of scheduled caste women shall have to be taken by deleting the corresponding number of candidates from the bottom of the list relating to Scheduled Castes. To this extent, horizontal (special) reservation differs from vertical (social) reservation. Thus women selected on merit within the vertical reservation quota will be counted against the horizontal reservation for women.
\end{indquote}

\subsection{Anil Kumar Gupta (1995): Implementation of Horizontal Reservations Compartmentalized within Vertical Reservations} \label{sec:AKG95}

While horizontal reservations can be implemented
either as \textit{overall horizontal reservations\/}  for the entire set of positions, or
as \textit{compartment-wise horizontal reservations} within each vertical category including the open category (OC),
the Supreme Court recommended the latter in their judgment  of \textit{Anil Kumar Gupta (1995)\/}:\smallskip
\begin{indquote}
We are of the opinion that in the interest of avoiding any
complications and intractable problems, it would be better that in future the horizontal reservations are
comparmentalised in the sense explained above. In other words, the notification inviting applications
should itself state not only the percentage of horizontal reservation(s) but should also specify the
number of seats reserved for them in each of the social reservation categories, viz., S.T., S.C., O.B.C. and O.C.
\smallskip
\end{indquote}
The procedure to implement compartmentalized horizontal reservation is described in \textit{Anil Kumar Gupta (1995)\/} as follows:
\begin{indquote}
The proper and correct course is to first fill up the O.C. quota (50\%) on the basis of merit: then fill up each of the
social reservation quotas, i.e., S.C., S.T. and B.C; the third step would be to find out how many
candidates belonging to special reservations have been selected on the above basis. If the quota fixed
for horizontal reservations is already satisfied - in case it is an over-all horizontal reservation - no
further question arises. But if it is not so satisfied, the requisite number of special reservation
candidates shall have to be taken and adjusted/accommodated against their respective social
reservation categories by deleting the corresponding number of candidates therefrom. (If, however, it
is a case of compartmentalised horizontal reservation, then the process of verification and
adjustment/accommodation as stated above should be applied separately to each of the vertical
reservations.
\end{indquote}
The adjustment phase of the procedure for implementation of horizontal reservation is further elaborated in the Supreme Court
judgment \textit{Rajesh Kumar Daria (2007)\/} as follows:
\begin{indquote}
If 19 posts are reserved for SCs (of which the quota for women is four), 19 SC candidates shall have
to be first listed in accordance with merit, from out of the successful eligible candidates. If such list of
19 candidates contains four SC women candidates, then there is no need to disturb the list by including
any further SC women candidate. On the other hand, if the list of 19 SC candidates contains only two
woman candidates, then the next two SC woman candidates in accordance with merit, will have to be
included in the list and corresponding number of candidates from the bottom of such list shall have to
be deleted, so as to ensure that the final 19 selected SC candidates contain four women SC candidates.
[But if the list of 19 SC candidates contains more than four women candidates, selected on own merit,
all of them will continue in the list and there is no question of deleting the excess women candidate on
the ground that `SC-women' have been selected in excess of the prescribed internal quota of four.]
\end{indquote}

\subsection{Saurav Yadav (2020): Revised Mandates on Implementation of Vertical and Horizontal Reservations \&
Indirect Enforcement of the 2SMG Choice Rule} \label{app-yadav}
Focusing on the model with non-overlapping HR protections,
the most visible mandates of the December 2020 Supreme Court judgment \textit{Saurav Yadav (2020)\/} in relation to our axioms are,
\begin{enumerate}
\item the clarification that all individuals are to be considered for  open-category HR-protected positions, thus
enforcing the axiom of \textit{maximal accommodation of HR-protections\/}, and
\item the enforcement of \textit{no justified envy\/} as its primary mandate.
\end{enumerate}
The following quote from the judgment, however, is also critical, because it implies that
our axiom of \textit{compliance with VR protection\/} in Definition  \ref{def-VR}
is also enforced in its stronger form with Condition 3:  \smallskip
\begin{indquote}
36. Finally, we must say that the steps indicated by the High Court of Gujarat in para 56 of its
judgment in Tamannaben Ashokbhai Desai contemplate the correct and appropriate procedure for considering
and giving effect to both vertical and horizontal reservations. The illustration given by us deals with only one possible dimension.
There could be multiple such possibilities. Even going by the present illustration, the first female candidate allocated in the vertical column
for Scheduled Tribes may have secured higher position than the candidate at Serial No.64.
In that event said candidate must be shifted from the category of Scheduled Tribes
to Open / General category causing a resultant vacancy in the vertical column of Scheduled Tribes.
Such vacancy must then enure to the benefit of the candidate in the Waiting List for Scheduled Tribes -- Female.\smallskip
\end{indquote}
More specifically the quote formulates the mandate that a member of a reserve-eligible category
(Scheduled Tribes in the example) has to be considered for open-category HR-protected positions
(for women HR protections in the example) before using up a VR-protected position.
Apart from its enforcement of our axiom \textit{compliance with VR protections\/}, this quote also brings clarity
for the following defining characteristic of VR protections, originally formulated in \textit{Indra Sawhney (1992)\/}, 
in the presence of HR protections: \\

\begin{indquote}
It may well happen that some members belonging to, say Scheduled Castes get selected in the open competition field on the basis of their own merit;
they will not be counted against the quota reserved for Scheduled Castes; they will be treated as open competition candidates.\smallskip
\end{indquote}
Prior to \textit{Saurav Yadav (2020)\/},
a formal interpretation was never provided for the following question:
What does it mean to be \textit{selected in the open competition field on the basis of one's own merit\/} in the presence of HR protections?
For the case of non-overlapping horizontal reservations, this question is answered as follows in \textit{Saurav Yadav (2020)\/}:
Any individual who is entitled to an open position based on her merit score, \textit{including those who are entitled to one
due to the adjustments to accommodate the HR protections\/}, is considered as an individual  who is selected
in the open competition field on the basis of her own merit.

Since the fourth axiom \textit{non-wastefulness\/} has always been enforced since \textit{Indra Sawhney (1992)\/}, our Theorem \ref{thm:2smg}
implies that the 2SMG choice rule is the only mechanism that satisfies all mandates of \textit{Saurav Yadav (2020)\/} for 
in field applications with non-overlapping HR protections.

\subsection{Tamannaben Ashokbhai Desai  (2020): High Court Mandate on Adoption of the
2SMG Choice Rule in the State of Gujarat} \label{sec:Gujarat}

With its August 2020 High Court  judgment  \textit{Tamannaben Ashokbhai Desai  (2020)\/},
the two-step minimum guarantee choice rule (2SMG) is now mandated for allocation of state public jobs  in the State of Gujarat.
While the choice rule is given in the judgment only for a single horizontal trait (women),
it  is also well-defined and well-behaved for multiple (but non-overlapping) traits as
presented in Theorem \ref{thm:2smg}.
Originally introduced in \cite{sonyen19} prior to the judgment of the High Court of Gujarat in December 2020,
the 2SMG choice rule  is  endorsed by the Supreme Court judgment  \textit{Saurav Yadav (2020)\/} for
the entire country.\footnote{The  two-step minimum guarantee choice rule is referred to as $C^{hor}_{2s}$ in  \cite{sonyen19}.}
Paragraph 56 of the High Court of Gujarat judgment
\textit{Tamannaben Ashokbhai Desai  (2020)\/}  describes the mandated procedure  as follows:

\begin{indquote}
For the future guidance of the State Government, we would
like to explain the proper and correct method of implementing
horizontal reservation for women in a more lucid manner.\\
PROPER AND CORRECT METHOD OF IMPLEMENTING HORIZONTAL RESERVATION FOR WOMEN   \\

\mbox{} $\vdots$ \\


\noindent Step 1: Draw up a list of at least 100 candidates
(usually a list of more than 100 candidates is
prepared so that there is no shortfall of
appointees when some candidates don’t join
after offer) qualified to be selected in the
order of merit. This list will contain the
candidates belonging to all the aforesaid
categories.\\ \smallskip
\noindent  Step 2: From the aforesaid Step 1 List, draw up a list
of the first 51 candidates to fill up the OC quota
(51) on the basis of merit. This list of 51
candidates may include the candidates
belonging to SC, ST and SEBC. \\ \smallskip
\noindent  Step 3: Do a check for horizontal reservation in OC
quota. In the Step 2 List of OC category, if
there are 17 women (category does not matter),
women’s quota of 33\% is fulfilled. Nothing more
is to be done. If there is a shortfall of women
(say, only 10 women are available in the Step 2
List of OC category), 7 more women have to be
added. The way to do this is to, first, delete the
last 7 male candidates of the Step 2 List.
Thereafter, go down the Step 1 List after item
no. 51, and pick the first 7 women (category
does not matter). As soon as 7 such women
from Step 1 List are found, they are to be
brought up and added to the Step 2 List to
make up for the shortfall of 7 women. Now, the
33\% quota for OC women is fulfilled. List of OC
category is to be locked. Step 2 List list
becomes final.  \\ \smallskip
\noindent  Step 4: Move over to SCs. From the Step 1 List, after
item no. 51, draw up a list of 12 SC candidates
(male or female). These 12 would also include
all male SC candidates who got deleted from the
Step 2 List to make up for the shortfall of
women.  \\ \smallskip
\noindent  Step 5: Do a check for horizontal reservation in the
Step 4 List of SCs. If there are 4 SC women,
the quota of 33\% is complete. Nothing more is
to be done. If there is a shortfall of SC women
(say, only 2 women are available), 2 more
women have to be added. The way to do this is
to, first, delete the last 2 male SC candidates of
the Step 4 List and then to go down the Step 1
List after item no. 51, and pick the first 2 SC
women. As soon as 2 such SC women in Step 1
List are found, they are to be brought up and
added to the Step 4 List of SCs to make up for
the shortfall of SC women. Now, the 33\% quota
for SC women is fulfilled. List of SCs is to be
locked. Step 4 List becomes final. If 2 SC
women cannot be found till the last number in
the Step 1 List, these 2 vacancies are to be
filled up by SC men. If in case, SC men are also
wanting, the social reservation quota of SC is to
be carried forward to the next recruitment
unless there is a rule which permits conversion
of SC quota to OC.  \\ \smallskip
\noindent  Step 6: Repeat steps 4 and 5 for preparing list of STs.  \\ \smallskip
\noindent Step 7: Repeat steps 4 and 5 for preparing list of SEBCs.
\end{indquote}

\section{Documentation of Evidence from Indian Court Rulings on\\ Disruption Caused by the Flaws of the SCI-AKG Choice Rule} \label{sec:India}

In this section we present extensive evidence on the disarray caused by the shortcomings of the SCI-AKG choice rule  in India.
Much of our analysis,  the  High Court judgments presented Section \ref{sec:justifiedenvy},
and our policy recommendations parallel  the arguments and the decision of the December 2020  Supreme Court
judgment  \textit{Saurav Yadav v State of Uttar Pradesh (2020)\/}. Our entire analysis and policy recommendations predate this important judgment,
and it was already presented in  an earlier draft of this paper in \cite{sonyen19}.

\subsection{Litigations on the SCI-AKG Choice Rule}\label{sec:challenges}
As we have argued in Section \ref{subsec-flaws},  the SCI-AKG choice rule fails our axioms of \textit{no justified envy\/}.
Moreover, it also fails \textit{incentive compatibility\/} due to backward class candidates losing their open-category
HR protections upon claiming their VR protections by declaring their backward class status.

The failure of SCI-AKG choice rule to satisfy no  justified envy is
fairly straightforward to observe. All it takes is a rejected backward class candidate   to realize that
her merit score is higher than an accepted general-category candidate, even though she qualifies for
all the HR protections the  less-deserving (but still accepted) candidate qualifies for.
Since the primary role of the reservation policy is positive discrimination
for candidates with more vulnerable backgrounds, this situation is very counterintuitive, and it
often results in legal action.
Focusing on complications caused by either anomaly,
we next present several court cases to document how they handicap concurrent implementation of vertical and
horizontal reservation policies in India.

\subsubsection{\textbf{High Court Cases Related to Justified Envy}}  \label{sec:justifiedenvy}

The possibility of justified envy under the  SCI-AKG choice rule has resulted in numerous court
cases throughout India for more than two decades, and since the presence of justified envy in the system is highly implausible,
these legal challenges often result in controversial rulings.
In addition, there are also cases where authorities who implement a better-behaved version of the
choice rule, one that does not suffer from this shortcoming, are nonetheless challenged in court,
on the basis that their adopted choice rules differ from those mandated by the Supreme Court.
These court cases are not restricted to lower courts, and include several cases in state high courts.  Even at the level of state high courts,
the judgments on this issue are highly inconsistent, largely due to the disarray created by
the possibility of justified envy under the SCI-AKG choice rule.
We next present several representative cases from high courts:
\begin{enumerate}
\item \textbf{\textit{Rajeshwari vs State (Panchayati Raj Dep) Ors, 15 March, 2013\/},
Rajasthan High Court}.\footnote{The case is available at \url{https://indiankanoon.org/doc/128221069/}
(last  accessed on 03/07/2019).}
This case combines 120 petitions against the State  of Rajasthan where the petitioners seek legal action
from the Rajasthan High Court by a large number of petitioners against the state government,
on the basis that reserve category women are allowed to benefit from open-category horizontally reserved positions for women.
The high court ruled that the state is at fault, and it must abandon its choice rule, adopting  the one mandated by the Supreme Court.
The following quote is from a story published in The Times of India covering this court case:\footnote{The Times of India story is available at
\url{https://timesofindia.indiatimes.com/city/jaipur/Womens-seats-on-open-merit-cant-be-filled-from-SC/ST-quota-High-court/articleshow/19101277.cms}
(last  accessed on 03/07/2019).}
\smallskip
\begin{indquote}
In a judgment that would affect all recruitments in the state government, the Rajasthan high court has ruled that posts
reserved for women in the open/general category cannot be filled with women from reserved categories even if the latter are
placed higher on the merit list$\ldots$

Women candidates who contested for different positions in at least three government departments, including the panchayati raj,
education and medical, last year had challenged the government move to allow ``migration'' of reserved category women to fill
the open category seats. The positions applied for included that of teachers Grade-II and III, school lecturers, headmasters and
pharmacists.\smallskip
\end{indquote}
Ironically, while the High Court's decision is correct, it also means that the better-behaved version of the choice rule has to be abandoned by the state.

\item \textbf{\textit{Ashish Kumar Pandey And 24 Others vs State Of U.P. And 29 Others on 16 March, 2016},
Allahabad High Court}.\footnote{The case is available at
\url{https://indiankanoon.org/doc/74817661/} (last  accessed on 03/07/2019).}
This lawsuit was brought to Allahabad High Court by 25
petitioners, disputing the mechanism employed by the State of
Uttar Pradesh---the most populous state in India with more than 200 million residents---to apply the provisions of horizontal reservations
in their allocation of more than 4,000 civil police and platoon commander positions.
Of these positions, 27\%, 21\%, and 2\%  are each vertically reserved for members of Other Backward Classes
(OBC), Scheduled Castes (SC), and Scheduled Tribes (ST), respectively, and
20\%, 5\%, and 2\% are each horizontally reserved for women,
ex-servicemen, and dependents of freedom fighters, respectively.
While only 19 women are selected for open-category positions based on their merit scores,
the total number of female candidates is less than even the number of open-category
horizontally reserved positions for women,
and as such all remaining women are  selected.
However, instead of assigning them positions from their respective backward class categories (as it is mandated under the SCI-AKG choice rule),
all of them are assigned positions from the open category.
Similarly, backward class candidates are deemed eligible to use horizontal reservations for dependents of freedom fighters and ex-servicemen as well.
The counsel for the petitioners argues that not only did the State of U.P. make an error in its
implementation of horizontal reservations,
but also that the error was intentional. The following quote is from the court case:\smallskip
\begin{indquote}
Per contra, learned counsel appearing for the petitioners would submit that fallacy was committed by the Board deliberately,
and with malafide intention to deprive the meritorious candidates their rightful placement in the open category.
The candidates seeking horizontal reservations belonging to OBC and SC category were wrongly adjusted in the open category,
whereas, they ought to have been adjusted in their quota provided in respective social category.
The action of the Board is not only motivated, but purports to take forward the unwritten agenda of the State Government
to accommodate as many number of OBC/SC candidates in the open category.\smallskip
\end{indquote}
The judge sides with the petitioners, and rules that the State of  Uttar Pradesh must correct
its erroneous application of the provisions of horizontal reservations.
The judge further emphasizes that the State has played foul, stating:\smallskip
\begin{indquote}
There is merit in the submission of the learned counsel for the petitioners that the conduct of the members of the
Board appears not only mischievous but motivated to achieve a calculated agenda by deliberately keeping meritorious
candidates out of the select list. The Board and the officials involved in the recruitment process were fully
aware of the principle of horizontal reservations enshrined in Act, 1993 and Government Orders which were being
followed by them in previous selections of SICP and PC (PAC), but in the present selection they chose to adopt
a principle against their own Government Orders and the statutory provisions which were binding upon them...

I am constrained to hold that both the State and the Board have played fraud on the principles enshrined in the Constitution with regard to public appointment.\smallskip
\end{indquote}
What is especially surprising is that, despite the heavy tone of this judgment, the State goes on to appeal
in another Allahabad High Court case \textit{State Of U.P. And 2 Ors. vs Ashish Kumar Pandey And 58 Ors,
29 July, 2016},\footnote{The case is available at \url{https://indiankanoon.org/doc/71146861/}
(last accessed on 03/07/2019).}
in an effort to continue using its preferred method for implementing horizontal reservations.
Perhaps unsurprisingly, this appeal was denied by the High Court.

This particular case clearly illustrates that there is strong resistance in at least some of the states to
implementing the provisions of horizontal reservations as mandated under the SCI-AKG choice rule.
While this resistance most likely reflects the political nature of this debate, the arguments of
the counsel for the state to maintain their preferred mechanism
are mostly based on the presence of justified envy under the SCI-AKG choice rule.
The following quote from the appeal illustrates that this was the main argument used in their defense:\smallskip
\begin{indquote}
The arguments that have been advanced on behalf of State and private appellant with all vehemence that women candidates
irrespective of their social class i.e. SC/ST/OBC are entitled to make place for themselves in an open category on their inter-se
merit clearly gives an impression to us that State of U.P and its agents/servants and even the private appellants are totally
unaware of the distinction that has been time and again reiterated in between vertical reservation and horizontal reservation
and the way and manner in which the provision has to be pressed and brought into play.\smallskip
\end{indquote}

\item \textbf{\textit{Asha Ramnath Gholap vs President, District Selection Committee \& Ors. on March 3rd, 2016},
Bombay High Court}.\footnote{The case is
available at \url{https://indiankanoon.org/doc/178693513/} (last accessed on 03/08/2019).}
In this case, there are 23 pharmacist positions to be allocated; 13 of these positions are vertically reserved for backward classes
and the remaining 10 are open for all candidates. In the open category, 8 of the 10 positions are horizontally
reserved for various groups,
including 3 for women. The petitioner, Asha Ramnath Gholap, is an SC woman, and while there is one vertically reserved
position for SC candidates, there is no horizontally reserved position for SC women.
Under the SCI-AKG choice rule, she is not eligible for any of the horizontally reserved positions for women at the open category.
Nevertheless, she brings her case to the Bombay High Court based on an instance of justified envy,
described in the court records as follows:\smallskip
\begin{indquote}
 It is the contention of the petitioner that Respondent Nos. 4 \& 5 have received less marks than the petitioner
 and as such, both were not liable to be selected. The petitioner has, therefore, approached this court by invoking
 the writ jurisdiction of this court under Article 226 of the Constitution of India, seeking quashment of the select list to
 the extent it contains the names of Respondent Nos.4 and 5 against the seats reserved for the candidates belonging to open female category.
\end{indquote}
Under the federal law,
there is no merit to this argument, because  the SCI-AKG choice rule allows for justified envy.
However,  the judges side with the petitioner
on the basis that a candidate cannot be denied a position from the open category based on her backward class membership,
essentially ruling out the possibility of justified envy under a Supreme Court-mandated choice rule,
which is designed to allow for positive discrimination for vulnerable groups.\footnote{In a very similar Bombay High Court case
\textit{Rajani Shaileshkumar Khobragade ... vs The State Of Maharashtra And ... on 31 March, 2017\/}
where the petitioner filed a lawsuit based on another instance of justified envy, the judges of the same high court dismissed the petition.
This case is available at \url{https://indiankanoon.org/doc/7250640/}, last accessed on 03/09/2019. Indeed, there seem to be several conflicting
decisions at the Bombay High Court on this very issue, including a series of cases reported in a \textit{The Times of India\/} story dated 07/18/2018
``MPSC won't issue job letters till HC hears plea on quota issue'' available at
\url{https://timesofindia.indiatimes.com/city/aurangabad/mpsc-wont-issue-job-letters-till-hc-hears-plea-on-quota-issue/articleshow/65029505.cms}
(last accessed on 03/09/2019).}
Their justification is given in the court records as follows:\smallskip
\begin{indquote}
 We find the argument advanced as above to be fallacious. Once it is held that general category or open category takes in
 its sweep all candidates belonging to all categories irrespective of their caste, class or community or tribe,
 it is irrelevant whether the reservation provided is vertical or horizontal.
 There cannot be two interpretations of the words `open category' $\dots$\smallskip
 \end{indquote}

\item \textbf{\textit{Uday Sisode vs Home Department (Police) on 24 October, 2017\/}, Madhya Pradesh High Court}.\footnote{The case is available at
\url{https://indiankanoon.org/doc/196750337/} (last accessed on 03/08/2019).}
In another case parallel to that at the Bombay High Court,
the judges of the Madhya Pradesh High Court issued a questionable decision by siding with a petitioner who filed this lawsuit
based on another instance of justified envy.

\item \textbf{\textit{Smt. Megha Shetty vs State Of Raj. \& Anr on 26 July, 2013\/}, Rajasthan High Court}.\footnote{The case is available at
\url{https://indiankanoon.org/doc/78343251/} (last accessed on 10/08/2019).}
In contrast to  \textit{Asha Ramnath Gholap (2016)\/} and \textit{Uday Sisode (2017)\/}
where the judges  have been erroneous
siding with petitioners whose lawsuits are based on instances of justified envy, in this case a general category petitioner seeks legal action
against the state on the basis that several HR-protected open-category positions for women are allocated
to women from OBC who are not eligible for these positions (unless they receive it without invoking the benefits
of horizontal reservation). While all these OBC women have higher merit scores than the petitioner and the state has apparently used
a better behaved procedure, the petitioner's case has merit
because the SCI-AKG choice rule allows for justified envy in those situations.
In an earlier lawsuit, the petitioner's case was already declined by a single judge of the same court based on an
erroneous interpretation of \textit{Indra Sawhney (1992)\/}. The petitioner subsequently appeals this erroneous decision
and brings the case to a larger bench of the same court. However, the three judges side with the earlier judgment,
thus erroneously dismissing the appeal. Their decision is justified as follows:
\begin{indquote}
The outstanding and important feature to be noticed is that it is not the case of the appellant-petitioner
that she has obtained more marks than those 8 OBC (Woman) candidates, who have been appointed against the
posts meant for General Category (Woman), inasmuch as, while the appellant is at Serial No.184 in the merit list,
the last OBC (Woman) appointed is at Serial No.125 in the merit list.
The controversy raised by the appellant is required to be examined in the context and backdrop of these significant factual aspects.
\end{indquote}
As seen from this argument, many judges have difficulty perceiving that the Supreme Court-mandated
procedure could possibly allow for justified envy.

\item \textbf{\textit{Mukta Purohit \& Ors vs State \& Ors on 12 April, 2018\/}, Rajasthan High Court}.\footnote{The case is available at
\url{https://indiankanoon.org/doc/126738191/} (last accessed on 10/10/2019).}
In a case that mimics \textit{Smt. Megha Shetty (2013)\/}, judges of the Rajasthan High Court erroneously dismiss a petition
filed against the state that allowed  HR-protected open-category positions for women to be allocated
to women from reserved categories who are ineligible.
Indeed \textit{Smt. Megha Shetty (2013)\/} is used as a precedent in this judgment.

\item \textbf{\textit{Arpita Sahu vs The State Of Madhya Pradesh on 21 August, 2012\/}  Madhya Pradesh
High Court}.\footnote{The case is available at
\url{https://indiankanoon.org/doc/102792215/} (last accessed on 10/10/2019).} The petitioner files a lawsuit based on an instance of justified envy, however in contrast to \textit{Asha Ramnath Gholap (2016)},   and \textit{Uday Sisode (2017)\/},
the judges have correctly dismissed the petition in this case.
\end{enumerate}

\subsubsection{\textbf{Wrongful Implementation and Possible Misconduct}} \label{sec:wrongful}

It is bad enough that the Supreme Court-mandated  SCI-AKG choice rule is not incentive-compatible,
forcing some candidates to risk losing their open-category HR protections by claiming their VR protections.
To make matters worse, in some cases candidates are denied access to open-category
HR protections even when they do not submit their backward class status,
giving up their VR protections.
Therefore, even when the candidate applies for a position as a general-category candidate without claiming the benefits of VR protections,
the central planner processes the application as if the backward class status was claimed,
denying the candidate's eligibility for open-category HR protections.
The central planners are often able to do this because last names in India are, to a large extent,
indicative of a caste membership. This type of misconduct seems to be fairly widespread in some jurisdictions, and
it is the main cause of lawsuits in dozens of cases such as
the two Bombay High Court cases
\textit{Shilpa Sahebrao Kadam vs The State Of Maharashtra (2019)\/} and
\textit{Vinod Kadubal Rathod vs Maharashtra State Electricity (2017)\/}.\footnote{The cases are available at \url{https://indiankanoon.org/doc/89017459/}  and
\url{https://indiankanoon.org/doc/162611497/} (last accessed on 03/09/2019).}
Indeed, this type of misconduct is sometimes intentional and systematic.
The following statement is from  \textit{Shilpa Sahebrao Kadam  (2019)\/}:
\begin{indquote}
According to Respondent - Maharashtra Public Service Commission, in view of the Circular dated 13.08.2014, only the candidates belonging to open (Non-reserved) category can be considered for open horizontally reserved posts meaning thereby, the reserved category candidates cannot be considered for open horizontally reserved post. Reference is made to a communication issued by the Additional Chief Secretary (Service) of the State of Maharashtra dated 26.07.2017, whereunder it is prescribed that a female candidate belonging to any reserved category, even if tenders application form seeking employment as an open category candidate, the name of such candidate shall not be recommended for employment against a open category seat.
\end{indquote}
Moreover, not all decisions in these lawsuits are made in accordance with the
SCI-AKR choice rule,  which allows  candidates to forego their VR (or HR) protections.
This is the case both for the first and last lawsuit listed above.
For example, in the last lawsuit, two petitioners
each applied for a position without declaring their backward class membership, with the
intention to benefit from open-category HR protections. Following their application,
these petitioners were requested to provide their school leaving certificates, which provided
information on their backward class status. Upon receiving this information, the petitioners
were declined eligibility for open-category HR protections, even
though they never claimed their VR protections. Hence, they
filed the fourth lawsuit given above. Remarkably, their petition was declined on the basis of
their backward class membership. Here we have a case where the authorities  not only go to
great lengths to obtain the backward class membership of the candidates, and wrongfully decline their
eligibility for open category HR protections, but they also manage to get their lawsuits dismissed.
The mishandling of this case is consistent with the concerns indicated in the February 2006 issue of
\textit{The Inter-Regional Inequality Facility\/} policy brief:\footnote{The policy brief is available at
\url{https://www.odi.org/sites/odi.org.uk/files/odi-assets/publications-opinion-files/4080.pdf} (last accessed 03/09/2019).}\smallskip
\begin{indquote}
Another issue relates to the access of SCs and STs to the institutions of justice in seeking protection against discrimination.
Studies indicate that SCs and STs are generally faced with insurmountable obstacles in their efforts to seek justice in the event of discrimination.
The official statistics and primary survey data bring out this character of justice institutions. The data on
Civil Rights cases, for example, shows that only 1.6\% of the total cases registered in 1991 were convicted, and that this had fallen to 0.9\% in 2000.\smallskip
\end{indquote}

\subsubsection{\textbf{Loss of Access to HR protections without any Access to VR protections}}\label{subsec:nover}
The main  justification offered in various Supreme Court cases for denying backward class members their open-category HR protections is avoiding a
situation where an excessive number of positions are reserved for members of these classes.
In several cases, however, members of these classes are denied access to open-category HR protections
even when the number of VR-protected positions is zero for their reserve-eligible vertical category.
This is the case  in the following two lawsuits:
\begin{enumerate}
\item \textit{Tejaswini Raghunath Galande v. The Chairman, Maharashtra Public Service
Commission and Ors. on 23 January 2019\/},  Writ Petition Nos. 5397 of 2016 \& 5396 of 2016, High Court of Judicature at Bombay.\footnote{The case is available at
\url{https://www.casemine.com/judgment/in/5c713d919eff4312dfbb5900} (last accessed on 03/09/2019).}
\item Original Application No. 662/2016 dated 05.12.2017, Maharashtra Administrative  Tribunal, Mumbai.\footnote{The case is available at
\url{https://mat.maharashtra.gov.in/Site/Upload/Pdf/O.A.662\%20of\%202016.pdf} (last accessed on 03/09/2019).}
\end{enumerate}
In both  cases, while the petitioners claimed their VR protections,
there was no VR-protected position for their class. Yet in  both cases petitioners lost  their
open-category HR protections. In the first case,
the petitioners' lawsuit to benefit from open-category HR protections was initially declined
by a lower court, resulting in an appeal at the High Court. The lower court's decision
was overruled in the High Court, and her request was granted. On the other hand, the second petitioner's
similar request was declined by the Maharashtra Administrative  Tribunal. What is more
worrisome in the second case is that, while initially three positions were VR-protected
for the petitioner's backward class, after the petitioner's application these VR-protected positions were
withdrawn. Therefore, the candidate declared her backward class status, giving up her
open-category HR protection, presumably
to gain access to VR-protected positions set aside for her reserve-eligible class, only to learn that she
had given up her eligibility for nothing.

\section{Original Formulation of the AKG-SCI Choice Rule\\ and Its Equivalence to Our Formulation} \label{sec:AKG-HAS}

The mechanics for implementing HR protections is described in the two Supreme Court judgments \textit{Anil Kumar Gupta (1995)\/}
and \textit{Rajesh Kumar Daria (2007)\/}, and given in Section \ref{sec:AKG95}.
In the main body of the paper we used a simpler formulation of the SCI-AKG choice rule, that relies on its
relation to the minimum guarantee choice rule.
In this section of the Online Appendix we formulate the original description of the SCI-AKG choice rule and
prove its equivalence to our simpler formulation.

Both judgments describe the procedure for a single trait, although the procedure can be repeated
sequentially for each trait. In our description below, we  adhere to this straightforward extension of the
procedure.

As we argue in Section \ref{sec:ver}, implementing VR protections is straightforward in the absence of HR protections.
Open-category positions are allocated to the highest merit score candidates (across all categories) first, followed by the
positions at each reserve-eligible category to the highest merit score remaining candidates from these categories.
This is indeed the first step of the SCI-AKG choice rule. Once a tentative assignment is determined, the necessary
adjustments are subsequently made to implement HR protections,
first for the open-category positions, then for positions at each reserve-eligible category.
The adjustment process is repeated for each trait.

Formally, for a given category $v \in \calv$ of positions, let
a set of individuals $J \subseteq \cali^{v}$  who are tentatively assigned
to category-$v$ positions and  a set of individuals $K \subseteq \cali^v\setminus J$ who are
eligible for horizontal adjustments at category $v$ be such that
\begin{enumerate}
\item $|J| = q^v$ and
\item $\sigma(i) > \sigma(i')$ \quad for any $i \in J$ and $i'\in K$.
\end{enumerate}
Then, for a given processing sequence
$t^1, t^2, \ldots, t^{|\calt|}$ of traits, the horizontal adjustment process is carried out with the following procedure. \medskip

\begin{quote}
\textbf{AKG Horizontal Adjustment Subroutine (AKG-HAS)} \smallskip

\noindent \textbf{Step 1} \textbf{(Trait-{\boldmath$t^1$} adjustments)}:
Let $r_1$ be the number of individuals in $J$ with trait $t^1$. 

\textbf{\textit{Case 1.}} $r_1 \geq q^v_{t^1}$

Let $J^1$ be the set of $q^v_{t^1}$ individuals with the highest merit scores in $J$ with trait $t^1$. Finalize
their assignments as the recipients of trait-$t^1$ HR-protected positions within category $v$.
Proceed to Step 2.  \smallskip

\textbf{\textit{Case 2.}} $r_1 < q^v_{t^1}$

Let $J^1_{m}$ be the set of all individuals in $J$ with trait $t^1$. Let $s_1$ be the number of individuals in $K$ who have trait $t^1$.
Let $J^1_{h}$ be
\begin{itemize}
\item the set of $(q^v_{t^1} - |J^1_{m}|)$ individuals with the highest merit scores in $K$ who have trait $t^1$ if  $s_1 \geq q^v_{t^1} - |J^1_{m}|$, and
\item   the set of all individuals in $K$ who have trait $t^1$ if $s_1 < q^v_{t^1} - |J^1_{m}|$.
\end{itemize}
Let $J^1 = J^1_{m}  \cup J^1_{h}$ and finalize their assignments as the recipients of trait-$t^1$ HR-protected positions within category $v$.
Proceed to Step 2.  \medskip

\noindent \textbf{Step} {\boldmath $k\in \{2, \dots, |\calt|\}$} \textbf{(Trait-{\boldmath$t^k$} adjustments)}:
Let $r_k$ be the number of individuals in $J\setminus \bigcup_{\ell=1}^{k-1} J^{\ell}$ with trait $t^k$. \smallskip

\textbf{\textit{Case 1.}} $r_k \geq q^v_{t^k}$

Let $J^k$ be the set of $q^v_{t^k}$ individuals with the highest merit scores in  $J\setminus \bigcup_{\ell=1}^{k-1} J^{\ell}$ with trait $t^k$. Finalize
their assignments as the recipients of trait-$t^k$ HR-protected positions  within category $v$.
Proceed to Step 2.  \smallskip

\textbf{\textit{Case 2.}} $r_k < q^v_{t^k}$

Let $J^k_{m}$ be the set of all individuals in  $J\setminus \bigcup_{\ell=1}^{k-1} J^{\ell}$  with trait $t^k$.
Let $s_k$ be the number of individuals in $K\setminus \bigcup_{\ell=1}^{k-1} J^{\ell}$ with trait $t^k$.
Let $J^k_{h}$ be
\begin{itemize}
\item the set of $(q^v_{t^1} - |J^k_{m}|)$ individuals with the highest merit scores in $K\setminus \bigcup_{\ell=1}^{k-1} J^{\ell}$
who have trait $t^k$ if  $s_k \geq q^v_{t^k} - |J^k_{m}|$, and
\item   the set of all individuals in $K\setminus \bigcup_{\ell=1}^{k-1} J^{\ell}$
who have trait $t^k$ if $s_k < q^v_{t^k} - |J^k_{m}|$.
\end{itemize}
Let $J^k = J^k_{m}  \cup J^k_{h}$ and finalize their assignments as the recipients of trait-$t^k$ HR-protected positions  within category $v$.
Proceed to Step $k+1$.  \medskip

\noindent \textbf{Step} {\boldmath  $(|\calt|+1)$} \textbf{(Finalization of category-{\boldmath$v$} no-trait assignments)}:
 Let $J^0$ be the set of $\big(q^v - \sum_{\ell=1}^{|\calt|} |J^{\ell}|\big)$ individuals with the highest merit scores in  $J\setminus \bigcup_{\ell=1}^{|\calt|} J^{\ell}$\medskip

The procedure selects the set of individuals in  $\bigcup_{\ell=0}^{|\calt|} J^{\ell}$.
 Here $J^0$ is the set of individuals from the original group $J$ who have survived the horizontal adjustment process
 without invoking any HR protection, and $J^k$  is the set of individuals
 who accommodate trait-$t^k$ HR protections for any trait $t^k$.\footnote{While all individuals in $J^k$ accommodate trait-$t^k$ HR protections,
only those who are in the set  $J^k \setminus J$ owe their assignments to trait-$t^k$ HR protections.} \medskip
\end{quote}

When each individual has at most one trait, it is easy to see that
the processing sequence of traits becomes immaterial under
the AKG-HAS, and it produces the same outcome as the category-$v$
\textit{minimum guarantee choice rule\/} $C^v_{mg}$.
The next result formulates this observation.

\begin{proposition} \label{prop:mg=has}
Suppose that each individual has at most one trait. Let $v\in \calv$ be any category of positions,
$J\subseteq \cali^{v}$ be a set of individuals who are tentatively assigned category-$v$ positions, and
$K\subseteq \cali^{v}\setminus J$ be a set of unmatched individuals who are eligible for category-$v$ positions.
If $|J|=q^v$ and every individual in $J$ has a higher merit score than every individual in $K$,
then  $C^v_{mg}(J\cup K)$ is the set of individuals who are assigned to category-$v$ positions under the
AKG-HAS.
\end{proposition}

\subsection*{Proof of Proposition \ref{prop:mg=has}}
Let $I=J\cup K$ and $I'$ be the set of individuals assigned to category-$v$
positions by AKG-HAS. We first show that
\begin{enumerate}
  \item $|I'|=q^v$,
  \item there exists no instance of justified envy involving an individual in $I'$ and an individual in $I\setminus I'$,
  \item $I'$ maximally accommodates category-$v$ HR protections for $I$.
\end{enumerate}
The proof then follows from Theorem \ref{thm:envchar} together with the equivalence of
the meritorious horizontal choice rule and the minimum guarantee choice rule when each individual has
at most one trait.

Proof of (1): $|I'|=q^v$ follows because at Step $|\calt|+1$ of AKG-HAS all positions are filled.

Proof of (2): Let $i\in I'$ and $j\in I\setminus I'$ such that $\sigma(j)>\sigma(i)$. Since $j\notin I'$, either $j$
does not have a trait or there are at least $q^v_t$ individuals in $I'$ where $t$ is $j$'s only trait. If $j$ does not
have a trait, then $i$ must have a trait $t'$ such that the number of individuals in $I'$ who have trait $t'$ is
$\min\{q^v_{t'},|\{i'\in I:t'\in \tau(i')\}|\}$. Then $n^v((I'\setminus \{i\}) \cup \{j\})=n^v(I')-1$, which means
that there is no instance of justified envy involving $j$ and $i$. If $j$ has trait $t$, then it must be that $i$ does
not have trait $t$, there are at least $q^v_t$ individuals with trait $t$ in $I'$, and $i$ must have a
trait $t' \neq t$ such that the number of individuals in $I'$ who have trait $t'$ is
$\min\{q^v_{t'},|\{i' \in I : t'\in \tau(i')\}|\}$. Then, as before, $n^v((I'\setminus \{i\}) \cup \{j\})=n^v(I')-1$, which
means that there is no instance of justified envy involving $j$ and $i$.

Proof of (3): For every trait $t$, there is a corresponding step of AKG-HAS so that the number of individuals
in $I'$ who have trait $t$ is $\min\{q^v_t,|\{i'\in I:t\in \tau(i')\}|\}$. Since each individual has at most
one trait, this implies that $n^v(I')=n^v(I)$.
\qed
\medskip

We are ready to present the original formulation of the SCI-AKG choice rule as it is described in
\textit{Anil Kumar Gupta (1995)\/} and \textit{Rajesh Kumar Daria (2007)\/}. Proposition \ref{prop:mg=has} presented
above immediately establishes the equivalence of the original formulation with the formulation presented in Section \ref{subsec-flaws}.

\begin{quote}
        \noindent{}{\bf SCI-AKG Choice Rule } {\boldmath$C^{SCI}$} \smallskip

\noindent For every $I \subseteq \cali$, \smallskip

        \noindent{}{\bf Step 1 (Open-category tentative assignment)}:
        \begin{itemize}
        \item If $|I| \leq q^o$ then assign  all individuals in $I$ to open-category positions and terminate the procedure.
        In this case $C^{SCI,o}(I) = I$ and $C^{SCI,c}(I) = \emptyset$ for any reserve-eligible category $c\in\calr$.
        \item Otherwise,  if $|I| > q^o$, then tentatively assign the highest merit-score $q^o$ individuals in $I$ to open-category positions.
Let $J^o$ denote the set of individuals who are tentatively assigned to open-category positions in this case.
\end{itemize}
 \medskip

        \noindent{}{\bf Step 2 (Finalization of open-category positions)}:
        The set of individuals eligible for open-category horizontal adjustments is $I^g\setminus J^o$.
        Apply the AKG-HAS
        \begin{itemize}
        \item to the set $J^o$ of tentative recipients of open-category positions
        \item with the set of individuals in $I^g\setminus J^o$ who are eligible for open-category horizontal adjustments
        \end{itemize}
        to finalize the set of recipients  $C^{SCI,o}(I)$  of open-category positions. \medskip

        \noindent{}{\bf Step 3  (Reserve-eligible category tentative assignment)}: For any reserve-eligible category $c \in \calr$,
        \begin{itemize}
        \item If $|I^c\setminus C^{SCI,o}(I)| \leq q^c$ then assign  all individuals in $I^c\setminus C^{SCI,o}(I)$ to category-$c$ positions,
        finalizing the assignments of  individuals in $I^c$.
        In this case $C^{SCI,c}(I) = I^c\setminus C^{SCI,o}(I)$.
        \item Otherwise,  if  $|I^c\setminus C^{SCI,o}(I)| > q^c$, then tentatively assign the highest merit-score $q^c$ individuals in
        $I^c\setminus C^{SCI,o}(I)$ to category-$c$ positions.
Let $J^c$ denote the set of individuals who are tentatively assigned to category-$c$ positions in this case.
       \end{itemize}
\medskip

        \noindent{}{\bf Step 4 (Finalization of reserve-eligible category positions)}:
        For any reserve-eligible category $c \in \calr$,
        the set of individuals eligible for category-$c$ horizontal adjustments is $I^c\setminus(C^{SCI,o}(I) \cup J^c)$.
        For any reserve-eligible category $c \in \calr$,  apply the AKG-HAS
        \begin{itemize}
        \item to the set $J^c$ of tentative recipients of category-$c$ positions
        \item with the set of individuals in $I^c\setminus(C^{SCI,o}(I) \cup J^c)$ who are eligible for category-$c$ horizontal adjustments
        \end{itemize}
        to finalize the set of recipients $C^{SCI,c}(I)$ of category-$c$ positions. \medskip

\noindent The outcome of the SCI-AKG choice rule  is $C^{SCI}(I) = \big(C^{SCI,v}(I)\big)_{v\in \calv}$.
\end{quote}

\end{document}